\documentclass[%
 reprint,
superscriptaddress,
 amsmath,amssymb,
 aps,
]{revtex4-2}

\usepackage{float}
\makeatletter
\let\newfloat\newfloat@ltx
\makeatother
\usepackage{graphicx}
\usepackage{graphicx}
\usepackage{dcolumn}
\usepackage{bm}
\usepackage{graphicx}
\usepackage{amsmath}
\usepackage{amssymb}
\usepackage{subcaption}
\usepackage{subfig}
\usepackage{qcircuit}
\usepackage{algorithm}
\usepackage{algpseudocode}
\usepackage{cprotect}
\usepackage{lipsum}
\usepackage{booktabs}
\usepackage{array}
\usepackage{comment}
\usepackage{amsmath, amsthm, amssymb}
\usepackage{braket}
\usepackage{bm}
\usepackage{physics}
\usepackage{graphics}
\usepackage[dvipsnames]{xcolor}

\newtheorem{proposition}{Proposition}
\theoremstyle{definition}

\makeatletter
\@ifundefined{sf@counterlist}{\newcommand{\sf@counterlist}{}}{}
\makeatother
\usepackage[colorlinks=true, allcolors=blue]{hyperref}

\usepackage{cleveref}

\begin{document}

\title{Improving Variational Quantum Circuit Optimization via Hybrid Algorithms and Random Axis Initialization}

\author{Joona V. Pankkonen}
\email{Corresponding author: joona.pankkonen@aalto.fi}
 \affiliation{%
 Micro and Quantum Systems group, Department of Electronics and Nanoengineering,\\Aalto University, Finland
}%
\author{Lauri Ylinen}
 \affiliation{%
 Department of Mathematics and Statistics,\\University of Jyväskylä, Finland
}%
\author{Matti Raasakka}
 \affiliation{%
 Micro and Quantum Systems group, Department of Electronics and Nanoengineering,\\Aalto University, Finland
}%

\author{Ilkka Tittonen}
 \affiliation{%
 Micro and Quantum Systems group, Department of Electronics and Nanoengineering,\\Aalto University, Finland
}%

\date{\today}

\begin{abstract}
Variational quantum circuits (VQCs) are an essential tool in applying noisy intermediate-scale quantum computers to practical problems. VQCs are used as a central component in many algorithms, for example, in quantum machine learning, optimization, and quantum chemistry. Several methods have been developed to optimize VQCs. In this work, we enhance the performance of the well-known \verb|Rotosolve| method, a gradient-free optimization algorithm specifically designed for VQCs. We develop two hybrid algorithms that combine an improved version of \verb|Rotosolve| with the free quaternion selection (\verb|FQS|) algorithm, which is the main focus of this study. Through numerical simulations, we observe that these hybrid algorithms achieve higher accuracy and better average performance across different ansatz circuit sizes and cost functions. We identify a trade-off between the expressivity of the variational ansatz and the speed of convergence to the optimum: a more expressive ansatz ultimately reaches a closer approximation to the true minimum, but at the cost of requiring more circuit evaluations for convergence. By combining the less expressive but fast-converging \verb|Rotosolve| with the more expressive \verb|FQS|, we construct hybrid algorithms that benefit from the rapid initial convergence of \verb|Rotosolve| while leveraging the superior expressivity of \verb|FQS|. As a result, these hybrid approaches outperform either method used independently.
\end{abstract}

\maketitle


\section{Introduction}
Hybrid quantum-classical variational algorithms are a promising tool to take advantage of the computational power of noisy intermediate-scale quantum (NISQ) computers. A computational advantage compared to the classical algorithms can be expected, in particular, for optimization problems and simulation of many-body quantum systems~\cite{Preskill2018quantumcomputingin}. However, there are several challenges with NISQ devices. In order to harness their full potential, error mitigation~\cite{quantum_error_mitigation1, quantum_error_mitigation2} and noise resistant algorithms~\cite{noise_resilience_algo1} need to be improved until we develop fault-tolerant quantum computing utilizing quantum error correction~\cite{quantum_error_correction}. 

Variational quantum circuits (VQCs) are the main building block of variational quantum algorithms~\cite{cerezo2021variational}. VQCs consist of a sequence of parameterized one- or two-qubit gates. For example, Variational Quantum Eigensolver~\cite{Kandala_2017} has been used for the Heisenberg antiferromagnetic model in an external magnetic field to find the ground state of the system. Other applications for VQCs range from quantum chemistry~\cite{quantum_chemistry}, hybrid quantum-classical machine learning~\cite{qml_reinforcement_learning, qml_generative, qml_qcnn}, and several optimization problems. In many cases, the parameter gradients of VQCs have been shown to vanish exponentially in the size of the ansatz circuit, the so-called barren plateau problem, which poses a serious challenge to the optimization of VQCs~\cite{cerezo2021variational}. However, these challenges can be mitigated to some extent with smart initialization techniques such as Gaussian initialization~\cite{gaussian_init}, Bayesian Learning initialization~\cite{bayesian_learning_init}, and Flexible Initializer for arbitrarily-sized Parametrized quantum circuits (FLIP)~\cite{flip_init}.

There is no single definite approach or algorithm to optimize VQCs. Several algorithms have been developed solely to optimize VQCs. In addition to the usual gradient-based methods, gradient-free methods have also been developed, such as \verb|Rotosolve| and \verb|Rotoselect| in~\cite{Ostaszewski_2021}, Free-Axis Selection (\verb|Fraxis|) in~\cite{fraxis}, and Free Quaternion Selection for quantum simulation (\verb|FQS|) in~\cite{Wada_2022}. In this work, we improve the performance of the well-known \verb|Rotosolve| method, which is a gradient-free optimization algorithm designed specifically for VQCs. We also create two hybrid algorithms consisting of the newly improved \verb|Rotosolve| version and \verb|FQS| algorithms, which are our main focus in this work.

First, we develop the modified version of \verb|Rotosolve|, \verb|Rotosolve-Haar|, where the generators of single-qubit gates are randomly sampled from the uniform distribution. Then, we combine \verb|Rotosolve-Haar| and \verb|FQS| into hybrid algorithms. Here, we observe a trade-off between the expressivity of the variational ansatz and the speed of convergence to the optimum: The more expressive the ansatz, the closer to the true minimum we get eventually, but the convergence to the optimum requires more circuit evaluations. By combining the two algorithms, the less expressive \verb|Rotosolve-Haar| and the more expressive \verb|FQS|, we obtain a family of hybrid algorithms, which combine the fast initial convergence of \verb|Rotosolve-Haar| with the better expressivity of \verb|FQS|, thus performing better than either one alone.

We investigate the performance of the algorithms using numerical simulations, considering the 5-, 6-, and 10-qubit one-dimensional Heisenberg model Hamiltonian, the 6- and 15-qubit two-dimensional Heisenberg model on a rectangular lattice, and the 4-qubit Hydrogen (H$_2$) molecular Hamiltonian. We also experiment with the algorithms' convergence to the randomly sampled $n$-qubit state with 4 and 5 qubits.

This work is structured as follows. In Sec.~\ref{VQC_section}, we go through the optimization of the variational quantum circuits, as well as \verb|Rotosolve| and \verb|FQS| methods. Then, in Sec.~\ref{hybrid_algo_section} we introduce the randomized version of \verb|Rotosolve|, termed \verb|Rotosolve-Haar|, and finally the proposed hybrid algorithms. In Sec.~\ref{results_section}, we show our results on one- and two-dimensional Heisenberg models as well as the hydrogen molecular Hamiltonian. Then, we present the hybrid algorithms' ability to converge to the randomly sampled $n$-qubit state in shallow circuits. Finally, in Sec.~\ref{section_conclusion} we conclude our work and discuss possible further research ideas.

\section{Optimization of variational quantum circuits} \label{VQC_section}

\subsection{Variational quantum circuits}
A variational quantum circuit (VQC) defines a unitary $U(\boldsymbol{\theta})$, which is parametrized by a set of parameters $\boldsymbol{\theta}$~\cite{cerezo2021variational}. Typically, the unitary operator corresponding to a variational quantum circuit with $L$ layers can be expressed as the product of the $L$ unitary operators as
\begin{equation}\label{unitaries}
    U(\boldsymbol{\theta}) = U_L(\boldsymbol{\theta}_L) \cdots U_2(\boldsymbol{\theta}_2) U_1(\boldsymbol{\theta}_1) \,.
\end{equation}
Denoting the number of qubits by $n$, the unitary $U_l(\boldsymbol{\theta}_l)$ corresponding to an individual layer may be expressed as
\begin{equation}
    U_l(\boldsymbol{\theta}_l) =  W_l \left( \bigotimes_{k=1}^n e^{-i \theta_{l,k} H_{l,k} / 2} \right) \,,
\end{equation}
where the index $k$ runs over the individual qubits. Here, $\theta_{l,k}$ is $k$-th element in parameter vector $\boldsymbol{\theta}_l$, $H_{l,k}$ is a single-qubit Hermitian operator generating the unitary $e^{-i \theta_{l,k} H_{l,k} / 2}$ which acts on the $k$-th qubit, and $W_l$ is a layer of entangling gates (e.g. CNOTs or controlled-Z gates). See Fig.~\ref{circuit1} for an example.

The main idea behind using VQCs is to minimize the objective or cost function with respect to the parameters $\bm{\theta}_l$. The cost function provides a quantitative measure of how well the given VQC performs on the given task. In general, a cost function $C(\boldsymbol{\theta})$ for a VQC is expressed as the expectation value of a Hermitian observable $\hat{M}$, obtained as
\begin{equation}\label{trace_expval}
    \langle \hat{M} \rangle = \text{Tr} \left( \hat{M} U(\boldsymbol{\theta}) \rho_0 U(\boldsymbol{\theta})^\dagger \right), \\[0.2cm]
\end{equation}
where $\rho_0 = \vert 0\rangle\langle0 \vert^{\otimes n}$ represents the initial state of the circuit.

\begin{figure}
    \[
    \Qcircuit @C=0.95em @R=0.9em {
    & \mbox{$L$ layers} & & & & \\
    & & & & & & & \\
     \lstick{\ket{0}_1} & \gate{R_{1}(\theta_{1})}  &  \ctrl{0} & \qw & \qw & \cdots & &  \gate{R_{Ln-4}(\theta_{Ln-4})} &  \ctrl{0} & \qw & \meter\\
     \lstick{\ket{0}_2}&  \gate{R_{2}(\theta_{2})} &  \ctrl{-1} & \ctrl{1}  & \qw & \cdots & &  \gate{R_{Ln-3}(\theta_{Ln-3})} &  \ctrl{-1} & \ctrl{1} & \meter\\
     \lstick{\ket{0}_3}&  \gate{R_{3}(\theta_{3})}  &  \ctrl{0} & \ctrl{0} & \qw &  \cdots & &  \gate{R_{Ln-2}(\theta_{Ln-2})} &  \ctrl{0} & \ctrl{0} & \meter\\
     \lstick{\ket{0}_4}&  \gate{R_{4}(\theta_{4})}  &  \ctrl{-1} & \ctrl{0} & \qw &  \cdots & &  \gate{R_{Ln-1}(\theta_{Ln-1})} &  \ctrl{-1} & \ctrl{0} & \meter\\
     \lstick{\ket{0}_5}&  \gate{R_{5}(\theta_{5})} &  \qw & \ctrl{-1} & \qw &  \cdots & &  \gate{R_{Ln}(\theta_{Ln})} & \qw & \ctrl{-1} & \meter  \gategroup{3}{2}{7}{4}{1.2em}{--} 
    }
    \]
    \cprotect\caption{An example of the circuit ansatz architecture for the 5-qubit structure optimization. Here \newline $R_i \in \{R_X, R_Y, R_Z \}$ is a single-qubit rotation gate and $\theta_i$ the corresponding angle parameter for $i=1,\ldots,Ln$.}
    \label{circuit1}
\end{figure}

\subsection{Rotosolve} \label{Rotosolve_section}

First, we assume that each layer has one parameterized single-qubit gate for each qubit. Then the parameter vector $\bm{\theta}$ consists of a total of $Ln$ optimizable parameters $\bm{\theta} = (\theta_1, \ldots \theta_{Ln-1}, \theta_{Ln})$, where $L$ and $n$ denote the number of layers and qubits in the circuit, respectively. We represent the $l$-th layer of fixed two-qubit gates (e.g., CNOTs or controlled-Z gates) as $W_l$. 

The $d$-th single-qubit gate $R_d$ ($d=1,\ldots,Ln$) is parameterized with a parameter $\theta_d \in \left(-\pi, \pi \right]$ as
\begin{equation}\label{rotosolve_gate}
    R_d(\theta_d) = \cos \left( \frac{\theta_d}{2} \right)I - i \sin \left( \frac{\theta_d}{2}\right) H_d,
\end{equation}\\
where $H_d$ is a Hermitian unitary generator with $H_d ^2 = I$. Each gate $R_d$ needs to be represented as a $2^n \times 2^n$ unitary matrix to act on the circuit. We achieve this by constructing a tensor product of a parameterized single-qubit gate $R_d$ that acts on a particular qubit and the $2\times 2$ identity matrices $I$ that act on all other qubits. With this in mind, we may express the unitary of the parameterized gate with the tensor product as follows

\begin{eqnarray} \label{rotosolve_gate}
    R_d' (\theta_d) =&& I \otimes \cdots \otimes I\otimes \left[ \cos\left( \frac{\theta_d}{2}\right)I - i \sin\left( \frac{\theta_d}{2}\right) H_d \right] \nonumber \\ && \otimes I \otimes\cdots \otimes I. 
\end{eqnarray}
 The expression for the unitary $U_l(\bm{\theta}_l)$ then becomes

\begin{equation}
    U_l(\bm{\theta}_l) = W_l R_{ln}'(\theta_{ln})R_{ln-1}'(\theta_{ln-1}) \cdots R_{(l-1)n +1}'(\theta_{(l-1)n+1}) ,
\end{equation}\\
where the index for $R_d'$ is determined from the layer index $l$ and the number of qubits $n$ to match the layer ansatz in Fig.~\ref{circuit1}.

Expressing the cost function $\hat{M}$ as a trace with initial state $\rho_0$ and the product of unitaries as shown in Eq.~(\ref{trace_expval}), the whole expression then becomes

\begin{widetext}
\begin{equation} \label{hat_M}
\begin{split}
        \langle \hat{M} \rangle = \text{Tr} & \Bigl( \hat{M} W_L  R_{Ln}'(\theta_{Ln})\cdots W_l R_{ln}'(\theta_{ln}) \cdots R_{d+1}'(\theta_{d+1})R_{d}'(\theta_{d}) R_{d-1}'(\theta_{d-1}) \cdots R_{2}'(\theta_{2}) R_{1}'(\theta_{1}) \\
        & \times \rho_0 R_1'(\theta_1)^\dagger R_{2}'(\theta_{2})^\dagger \cdots R_{d-1}'(\theta_{d-1})^\dagger R_{d}'(\theta_{d})^\dagger R_{d+1}'(\theta_{d+1})^\dagger\cdots R_{ln}'(\theta_{ln})^\dagger W_l^\dagger \cdots R_{Ln}'(\theta_{Ln})^\dagger W_L^\dagger\Bigr) \\[0.2cm]
\end{split}
\end{equation}
\end{widetext}
From here, we define the quantum circuits that come before and after the $d$-th parameterized gate $R_d$ as $A$ and $B$, respectively. Both $A$ and $B$ depend on the values $\theta_k$, where $k\neq d$. Then, the compressed form of the expectation value for $\langle\hat{M}\rangle$ reads
\begin{equation}\label{hat_M_short}
    \langle \hat{M} \rangle =  \text{Tr}\left(\hat{M} A R_d'(\theta_d) B \rho_0 B^\dagger R_d '(\theta_d)^\dagger A^\dagger\right).
\end{equation}
By using the cyclic property of the trace operation and defining 
\begin{align}
    \rho &\equiv B \rho_0 B^\dagger, \label{rho_M1}\\
    M & \equiv A^\dagger \hat{M} A, \label{rho_M2}
\end{align}

\noindent
we arrive at a simpler form of the expectation value of the cost function~\cite{Ostaszewski_2021, fraxis}
\begin{equation} \label{hermitian_observable}
    \langle M\rangle =  \text{Tr}\left( M R_d'(\theta_d) \rho R_d'(\theta)^\dagger\right).\\[0.2cm]
\end{equation}
By following the notation from Ref.~\cite{Ostaszewski_2021} the cost function in the above equation can be expressed w.r.t. parameters $\theta_d$ as 
\begin{eqnarray}
    \langle M\rangle_{\theta_d} &&=\cos^2 \left(\frac{\theta_d}{2} \right) \text{Tr} (M\rho) \nonumber \\[0.1cm]&&  + i \sin \left(\frac{\theta_d}{2} \right) \cos \left(\frac{\theta_d}{2} \right) \text{Tr}(M [\rho, H_d ']) \nonumber \\[0.1cm]&& +\sin^2 \left(\frac{\theta_d}{2} \right) \text{Tr}(MH_d '\rho H_d ').    
\end{eqnarray}
\\
Here we have defined $H_d' = I \otimes\cdots \otimes I \otimes H_d \otimes I \otimes \cdots \otimes I$ to be $2^n\times 2^n$ matrix to match the dimensions of $\rho$ and $M$. The traces can be computed using the expectation values by setting different values for $\theta_d$. The traces are evaluated for the angles $\theta_d = 0, \pi/2, -\pi/2, $ and $\pi$ as follows (see full details in~\cite{Ostaszewski_2021})
\begin{align} 
    \text{Tr} (M\rho) &= \langle M \rangle_0, \label{rotosolve_measurements1} \\[0.2cm]
    i \text{Tr}(M [\rho, H_d ']) &= \langle M \rangle_{\frac{\pi}{2}} -  \langle M \rangle_{-\frac{\pi}{2}}, \label{rotosolve_measurements2} \\[0.2cm]
    \text{Tr}(MH_d '\rho H_d ') &= \langle M \rangle_\pi. \label{rotosolve_measurements3} 
\end{align}
By using trigonometric identities, $\langle M \rangle_{\theta_d}$ can be expressed as a sinusoidal function $\langle M \rangle_{\theta_d} = A \sin(\theta_d + B) + C$, where the coefficients $A, B$ and $C$ are composed of traces from Eqs.~\eqref{rotosolve_measurements1}--\eqref{rotosolve_measurements3} 

\begin{align}
    A &= \frac{1}{2}\sqrt{(\langle M \rangle_{0} -  \langle M \rangle_{\pi})^2 + (\langle M \rangle_{\frac{\pi}{2}} -  \langle M \rangle_{-\frac{\pi}{2}})^2}, \\
    B&=\arctan2 \left( \langle M \rangle_{0} -  \langle M \rangle_{\pi}, \langle M \rangle_{\frac{\pi}{2}} -  \langle M \rangle_{-\frac{\pi}{2}}\right), \\
    C&= \frac{1}{2} \left( \langle M \rangle_{0} +  \langle M \rangle_{\pi}\right).
\end{align}\\
The minimum is located at $\theta_d^* = -\frac{\pi}{2} - B + 2\pi k$, where $k \in \mathbb{Z}$. At first glance, the coefficient $B$ requires 4 circuit evaluations in total to optimize the gate $R_d$. However, the number of circuit evaluations needed can be reduced to three by rewriting the coefficient $B$ with an arbitrary angle $\phi$ as (see Ref.~\cite{Ostaszewski_2021})

\begin{eqnarray}
    B &&= \arctan2 ( 2\langle M \rangle_{\phi} -  \langle M \rangle_{\phi + \frac{\pi}{2}} - \langle M \rangle_{\phi - \frac{\pi}{2}},\nonumber \\ && \quad \quad \quad \quad \quad \langle M \rangle_{\phi + \frac{\pi}{2}} - \langle M \rangle_{\phi - \frac{\pi}{2}} ) - \phi.
\end{eqnarray}

Finally, the optimal value of $\theta_d$ for the gate $R_d$ can be written as follows 

\begin{eqnarray}
    \theta_d ^* &&= \underset{\theta_d}{\arg \min} \langle M \rangle_{\theta_d} \nonumber \\
    &&= \phi - \frac{\pi}{2} - \arctan2 (2  \langle M \rangle_{\phi} -  \langle M \rangle_{\phi+ \frac{\pi}{2}} - \langle M \rangle_{\phi - \frac{\pi}{2}}, \nonumber \\ && \quad \quad \quad \quad \quad \quad \quad \quad  \langle M \rangle_{\phi+ \frac{\pi}{2}} - \langle M \rangle_{\phi - \frac{\pi}{2}} ) +2\pi k,
\end{eqnarray}\\
where $k \in \mathbb{Z}$. In this work, we set $\phi$ to zero but it can also be chosen arbitrarily.

This optimization of parameters $\theta_d$ is used in the \verb|Rotosolve| and \verb|Rotoselect| algorithms in~\cite{Ostaszewski_2021}. \verb|Rotosolve| initializes all parameterized gates at random and then optimizes one parameter at a time until all gates are optimized in the circuit. 

\subsection{Free Quaternion Selection} \label{fqs_section}

Now we explain how the free quaternion selection (\verb|FQS|) from Ref.~\cite{fqs_algo} works in general. \verb|FQS| is an extension of the \verb|Fraxis| algorithm~\cite{fraxis} which optimizes a three-dimensional axis with fixed rotational angle $\pi$. Both \verb|Fraxis| and \verb|FQS| are gradient-free sequential algorithms that are based on matrix diagonalization. The elements of the matrix are computed from the expectation values of the Hermitian observable $\hat M$. \verb|FQS| uses a quaternion selection of a single-qubit gate $R_{\bm{n}}(\psi)$ and allows the cost function to be rewritten as a solvable quadratic function, making optimization more efficient. In this section, we follow the notation used in Ref.~\cite{fqs_algo}.

Similarly to the previous section, we define a $2^n\times 2^n$ matrix, constructed as the tensor product of a parameterized single-qubit gate applied to a specific qubit and the identity matrix applied to all remaining qubits. To simplify the notation, we drop the subscript $d$ from the angle variables $\psi, \theta$ and $\phi$ used in this section. Then the unitary of the $d$-th single-qubit gate acting in the circuit is expressed as
\begin{eqnarray}\label{fqs_gate_def}
    R_{\bm{n}_d}'(\psi) &= I \otimes \cdots \otimes I\otimes \left[ \cos \left( \frac{\psi}{2} \right) I - i \sin \left( \frac{\psi}{2} \right) (\bm{n}_d\cdot \bm{\sigma}) \right] \nonumber \\& \otimes I \otimes \cdots \otimes I, 
\end{eqnarray}
where $I$ is the identity operator for single-qubit,  $\bm{\sigma}=(\sigma_x, \sigma_y, \sigma_z)$ is the 3-vector of Pauli matrices, and $\bm{n}_d$ is a unit 3-vector. By using a spherical coordinate system to express $\bm{n}_d$ with zenith angle $\theta$ and azimuth angle $\phi$, we have then 
\begin{equation}\label{spherical_coords}
    \bm{n}_d = \bm{n}_d(\theta, \phi) = (\cos \theta, \sin \theta \cos \phi, \sin \theta\sin \phi).
\end{equation}
By defining an extended vector of Pauli matrices as $\bm{\varsigma} \equiv (\varsigma_I, \varsigma_x, \varsigma_y, \varsigma_z) = (I, -i \sigma_x, -i \sigma_y, -i\sigma_z)$, we can write the single-qubit gate unitary with a unit quaternion as 
\begin{eqnarray} \label{quaternion_rep}
    R_{\bm{n}_d (\theta, \phi)}'(\psi) =&& I \otimes \cdots \otimes I \otimes \left[ \bm{q}_d(\psi, \theta, \phi)\cdot \bm{\varsigma} \right] \nonumber \\ &&\otimes I \otimes \cdots \otimes I \equiv R'(\bm{q}_d),
\end{eqnarray}
where the unit quaternion $\bm{q}_d = (q_0, q_1, q_2, q_3)$ can be expressed with angles $\psi, \theta$ and $\phi$ as follows:
\begin{align}
\begin{split}
  q_0 &= \cos \left( \frac{\psi}{2} \right), \\
  q_1 &= \sin \left( \frac{\psi}{2} \right) \cos \theta \\
  q_2 &= \sin \left( \frac{\psi}{2} \right) \sin \theta \cos \phi\\
  q_3 &= \sin \left( \frac{\psi}{2} \right) \sin \theta \sin \phi.
\end{split}
\end{align}

Now we focus on the optimization of the single-qubit gate $R(\bm{q}_d)$ which is defined without the tensor products with the identity matrices as in Eq.~(\ref{rotosolve_gate}) in the previous section. Likewise, in the \verb|Rotosolve| algorithm, we fix all gates except the $d$-th one. The expectation value for the Hermitian operator $\hat{M}$ has the exact same form as in Eq.~(\ref{hat_M}) in the previous section. Again, we define quantum circuits that come before and after the $d$-th parameterized gate $R_d$ as $A$ and $B$, respectively. Here, $A$ and $B$ depend on the quaternion values $\bm{q}_k,$ where $k \neq d$. Thus, we obtain the same form for $\langle\hat{M}\rangle$ as in Eq.~(\ref{hat_M_short}). Then, using the definitions from Eqs.~(\ref{rho_M1}) and~(\ref{rho_M2}) we have
\begin{equation}
    \langle M \rangle_{\bm{q}_d} = \text{Tr} \left(M R'( \bm{q}_d) \rho R'(\bm{q}_d)^\dagger \right).
\end{equation}
Using the quaternion representation from Eq.~(\ref{quaternion_rep}) we obtain the quadratic form for the expectation value
\begin{equation}
    \langle M \rangle_{\bm{q}_d} = \bm{q}_d ^T S \bm{q}_d .
\end{equation}
Here we have absorbed the Hermitian observable $M$ and the density matrix $\rho$ into the matrix $S = (S_{\mu\nu})$, which is a symmetric and real $4 \times 4$ matrix. The vector $\bm{q}_d^T$ is the transpose of $\bm{q}_d$, and the matrix elements of $S$ are defined as
\begin{equation}
    S_{\mu \nu} = \frac{1}{2} \text{Tr}\left[M \left( \varsigma_\mu' \rho \varsigma_\nu '^\dagger +  \varsigma_\nu' \rho \varsigma_\mu'^\dagger \right)\right].
\end{equation}
Here, we have defined $\varsigma_\mu ' = I \otimes \cdots \otimes I \otimes \varsigma_\mu \otimes I \otimes \cdots \otimes I$ as a $2^n \times 2^n$ matrix to match the dimensions of $\rho$ and $M$, similar to $H_d'$ in the previous section. 

To compute all the required matrix elements $S_{\mu \nu}$, a total of 10 circuit evaluations are needed, as it is symmetric and real. Each entry in the upper diagonal of the matrix corresponds to a specific combination of indices $\mu$ and $\nu$, which are classified into three types. See full details in Ref.~\cite{fqs_algo}. After forming the matrix $S$, the quadratic form of the expectation value $\langle M \rangle_{\bm{q}_d}$ is minimized by solving the eigenvector of $S$ corresponding to the lowest eigenvalue. This eigenvector corresponds to the optimal coefficients of the quaternion $\bm{q}_d$. Next, we show how to utilize a random initialization method on its own and with a combination of the \verb|FQS| algorithm.

\section{Random axis initialization method and hybrid algorithms}\label{hybrid_algo_section}


We intend to generalize single-qubit parameterized gates by applying a unitary transformation $V$ on a single-qubit gate generator $H_d = Z$ to obtain a new generator $G_d$ as follows:
\begin{equation}
    G_d = V H_d V^\dagger.
\end{equation}
This was originally suggested in Ref.~\cite{Ostaszewski_2021}. Obviously, the generator retains its Hermiticity and unitarity under conjugation with a unitary, so we still have $(G_d)^2 = I$. Therefore, we obtain a new single-qubit gate as $R_d(\theta_d) = \exp \left(-i \frac{\theta_d}{2}V_d H_d V_d ^\dagger \right)$, where $V_d$ is $d$-th randomly sampled instance from the Haar measure. We remark that the gate acting on multiple qubits can be defined analogously to Eqs.~(\ref{rotosolve_gate}) and~(\ref{fqs_gate_def}).

Now we propose a random axis initialization method based on the unitary transformation for single-qubit generators $H_d$. The randomized \verb|Rotosolve| algorithm is defined in Algorithm~\ref{roto_haar}. The Algorithm~\ref{roto_haar} in its core is the same as the original \verb|Rotosolve| algorithm presented in Ref.~\cite{Ostaszewski_2021}, except that the single-qubit gate generators are randomized at the start.  We may call this version of \verb|Rotosolve| from now on \verb|Rotosolve-Haar|. The random axis initialization provides a way to express more complex quantum states with the same ansatz, i.e. higher expressibility when dealing with deep VQCs. We demonstrate this in the following section. We emphasize that the representation of \verb|Rotosolve-Haar| enables us to combine \verb|Rotosolve| and \verb|FQS| into the hybrid algorithms, since they require a change of representation of parameters from one algorithm to another. That is, when switching from the quaternion representation used in \verb|FQS| to \verb|Rotosolve-Haar|, each quaternion must be converted into the corresponding unitary $V_d$ and rotation angle $\theta_d$ used by \verb|Rotosolve-Haar|, ensuring that the gates remain equivalent.

\begin{algorithm}
    \caption{Rotosolve-Haar (Haar-random unitary initialization for Rotosolve)}\label{roto_haar}
    \begin{algorithmic}[1]
    \State \textbf{Inputs}: A Variational Quantum Circuit $U$ with fixed architecture, Hermitian measurement operator as the cost function, and a stopping criterion.
    \State Initialize $\theta_d \in (-\pi, \pi]$ for $d = 1, \ldots, Ln$.
    \State Sample random unitaries $V_d$ from the Haar measure for $d = 1, \ldots, Ln$.
    \State Compute new generators $G_d = V_d Z V_d ^\dagger$ for $d = 1, \ldots, Ln$.
    \State Set $\phi = 0$.
    \Repeat
        \For{$d = 1, \ldots, Ln$} 
            \State Fix all angles except the $d$-th one
            \State Estimate $\langle M \rangle_\phi$, $\langle M \rangle_{\phi+\frac{\pi}{2}}$ and $\langle M \rangle_{\phi-\frac{\pi}{2}}$ from samples
            \State $\theta_d \gets \phi - \frac{\pi}{2} - \arctan2(2\langle M \rangle_\phi - \langle M \rangle_{\phi+\frac{\pi}{2}} - \langle M \rangle_{\phi-\frac{\pi}{2}}, \langle M \rangle_{\phi+\frac{\pi}{2}} - \langle M \rangle_{\phi-\frac{\pi}{2}})$
        \EndFor
    \Until{stopping criterion is met}
    \end{algorithmic}
\end{algorithm}

Finally, we present two hybrid algorithms consisting of \verb|Rotosolve-Haar| (Algorithm~\ref{roto_haar}) and \verb|FQS| from Ref.~\cite{fqs_algo}. The first hybrid algorithm is iteration-based, where we use algorithm $\mathcal{A}$ for one iteration of optimization and another algorithm $\mathcal{B}$ otherwise, and it is defined in Algorithm~\ref{hybrid_algorithm}. To simplify the explanation, we consider replacing algorithms $\mathcal{A}$ and $\mathcal{B}$ with \verb|FQS| and \verb|Rotosolve-Haar|, respectively. The algorithm starts by initiating the parameters in the same way as in the randomized version of the \verb|Rotosolve| algorithm, \verb|Rotosolve-Haar|. We also set a variable $N$ to a positive integer value, which we use to determine the algorithm we use in the current iteration. For every $N$-th iteration, we execute the \verb|FQS| algorithm for the whole iteration, otherwise, we use \verb|Rotosolve-Haar|. When we switch from \verb|Rotosolve-Haar| to \verb|FQS|, the parameters first need to be converted to a quaternion representation. During conversion to the quaternion form, we convert one gate at a time and extract the coefficients of $\bm{q}_d$ for all gates. After the \verb|FQS| optimization iteration, the parameters in quaternion form must be converted back and expressed as $\exp \left({-i \frac{\theta_d}{2} V_d H_d V_d ^{\dagger}} \right)$. To obtain the correct $V_d$ and $\theta_d$, we solve the equation $R(\bm{q}_d) = \exp \left({-i \frac{\theta_d}{2} V_d H_d V_d ^{\dagger}} \right)$ w.r.t. $\theta_d$ and $V_d$ for each $d$. We provide the proof that the angle $\theta$ and unitary $V$ are solvable from a given unitary $U$ in Appendix~\ref{appendix_rotosolve_haar_proof}. In the experiments, we also test a hybrid algorithm where we switch places of \verb|FQS| and \verb|Rotosolve-Haar|. That is, we use \verb|Rotosolve-Haar| every $N$-th iteration in the optimization, and otherwise \verb|FQS| is used.

\begin{algorithm}
\caption{Iteration-specific hybrid algorithm}\label{hybrid_algorithm}
\begin{algorithmic}[1]
\State \textbf{Inputs}: A Variational Quantum Circuit $U$ with fixed architecture, Hermitian measurement operator as the cost function, a stopping criterion, and two algorithms $\mathcal{A}$ and $\mathcal{B}$ (e.g. \verb|Rotosolve-Haar| and \verb|FQS|, respectively).
\State Initialize $\theta_d \in (-\pi, \pi]$ for $d = 1, \ldots, Ln$.
\State Set running index $i = 1$
\State Set an integer value to $N = k$, where $k \in \mathbb{Z}_+$
\Repeat
    \State $r \gets i \bmod N$
    \If {r = 0} 
        \For{$d = 1, \ldots, Ln$}
            \State Fix all gates except the $d$-th one
            \State Optimize $R_d$ with the algorithm $\mathcal{A}$

        \EndFor
            
    \Else  
    \For{$d = 1, \ldots, Ln$}
        \State Fix all gates except the $d$-th one 
        \State Optimize $R_d$ with the algorithm $\mathcal{B}$
    \EndFor
    \EndIf
    
    \State $i \gets i + 1$
\Until{stopping criterion is met}
\end{algorithmic}
\end{algorithm}

The second hybrid algorithm that we propose is a gate-specific hybrid algorithm of \verb|Rotosolve-Haar| and \verb|FQS|, which is defined in Algorithm~\ref{markov_chain_hybrid}. As the name suggests, we optimize the given gate \verb|Rotosolve-Haar| or \verb|FQS| in a probabilistic manner. In optimization, we fix all gates except the $d$-th one and with a probability $p$  we optimize it by using \verb|Rotosolve-Haar|, otherwise the gate is optimized by the \verb|FQS| algorithm. In each case, we transform the gate to the correct representation before performing the optimization. Note that the probability $p$ remains fixed during the whole optimization process.

\begin{algorithm}
\caption{Gate-specific hybrid algorithm}\label{markov_chain_hybrid}
\begin{algorithmic}[1]
\State Initialize all gates as done in \verb|Rotosolve-Haar| (Algorithm~\ref{roto_haar}).
\State Set a fixed value for probability threshold $p$
\Repeat 
        \For{$d = 1, \ldots, Ln$}
            \State Fix all gates except the $d$-th one
            \State $r \gets \text{Uniform}(0,1)$
            \If {$r < p$} 
                \State Find $\theta_d\in \mathbb{R}$ and $V_d\in U(2)$ such that $U_d=e^{-i\frac{\theta_d}{2} V_d H_d V_d^\dagger}$ (see Appendix~\ref{appendix_rotosolve_haar_proof})
                \State Optimize $e^{-i\frac{\theta_d}{2} V_d H_d V_d^\dagger}$ over $\theta_d$ with \verb|Rotosolve|

            \Else 
                \State Optimize $R_d$ with \verb|FQS|
            \EndIf
        \EndFor
            
\Until{stopping criterion is met}
\end{algorithmic}
\end{algorithm}

As a stopping criterion, we use the number of circuit evaluations as a metric for both hybrid algorithms. In the case of \verb|Rotosolve-Haar| gates, 3 circuit evaluations are needed and 10 for \verb|FQS| gates, respectively. In this work, we also compare our methods with \verb|Fraxis|, which requires 6 circuit evaluations to optimize the single-qubit gate. We focus on experimenting with our methods as a function of circuit evaluations in all experiments to obtain a detailed comparison of algorithms. A collection of different algorithms, as well as the number of circuit evaluations needed and their parameter used, is provided in Table~\ref{algo_table}. Here, the circuit evaluations needed on average with hybrid algorithms are expressed as an expectation value for one gate. In Table~\ref{parameter_values_table}, we provide concrete numbers for the average number of circuit evaluations for each gate while using iteration- and gate-specific hybrids with different values of $p$ and $N$.

\begin{table*}
\centering
\scalebox{0.9}{%
\begin{tabular}{|c |c | c | c | c | c | c | c |} 
 \hline
 Algorithm & Rotosolve & Rotosolve--Haar & Fraxis & FQS & Gate-specific hybrid & FQS every $N$  & Rotosolve--Haar every $N$ \\ [0.2ex] 
 \hline\hline
 Optimization variable & Angle $\theta$ & Angle $\theta$ & Axis $\bm{n}$ & Quaternion $\bm{q}$ & $\theta$ or $\bm{q}$ & $\theta$ or $\bm{q}$ & $\theta$ or $\bm{q}$  \\ 
 \hline
 Circuit evaluations & 3 & 3 & 6 & 10 & $3p + 10(1-p)$ & $[3(N-1) + 10]/N$ & $[10(N-1) + 3]/N$ \\
 \hline
\end{tabular}}
\caption{A collection of different algorithms and their optimization methods.} 
\label{algo_table}
\end{table*}

\begin{table*}
\centering
\scalebox{1.2}{%
\begin{tabular}{|c |c | c | c | c | c | c | c | c| c | c |} 
\hline
\multicolumn{1}{|c|}{Algorithm} & 
\multicolumn{4}{|c|}{Gate-specific hybrid} &
\multicolumn{3}{|c|}{FQS every $N$} &
\multicolumn{3}{|c|}{Rotosolve--Haar every $N$} \\
\hline
 Parameter value & $p=0.2$ & $p=0.4$ & $p=0.6$ & $p=0.8$ & $N=2$ & $N=3$ & $N=4$ & $N=3$ & $N=4$ & $N=5$ \\ 
 \hline
 Circuit evaluations & 8.6 & 7.2 & 5.8 & 4.4 & 6.5 & 5.33 & 4.75 & 7.67 & 8.25 & 8.6 \\
 \hline
\end{tabular}}
\caption{Circuit evaluations on average needed to optimize one gate while using the hybrid methods with different parameter values of $p$ and $N$.} 
\label{parameter_values_table}
\end{table*}

\section{Results} \label{results_section}

In this section, we examine the results of the different Hamiltonians and algorithms presented in the previous section. We start by considering a Hamiltonian for the Heisenberg model in the one-dimensional cyclic lattice for 5, 6, and 10 qubits. In addition to the one-dimensional Heisenberg model, we optimized the two-dimensional Heisenberg Hamiltonian in a rectangular lattice with the nearest neighbor connectivity and open boundary conditions for 6 and 15 qubits. The lattice sizes were set to $2 \times 3$ and $3\times 5$, respectively. Then, we show our results for the 4-qubit hydrogen molecular Hamiltonian. Finally, we compare the expressibility between the algorithms in terms of random state optimization. The trace distance between the optimized state and the randomly sampled target state is measured as a function of circuit evaluations. In all experiments, we compare the performance of the gradient-based algorithm, Adam~\cite{adam_opt}, with the gradient-free algorithms \verb|Rotosolve-Haar|, \verb|Rotosolve|, \verb|Fraxis|, \verb|FQS|, as well as hybrid algorithms.

The circuits in all the following results were prepared in the same way, layer by layer. One layer consists of single-qubit gates, where we assign one gate for each qubit, followed by a ladder of controlled-Z gates to strongly entangle the qubits. The structure of the layers was repeated $L$ times to create the circuit. This is similar to the circuit ansatz as done in Ref.~\cite{McClean_2018}. An illustration of the ansatz circuit is shown in Fig.~\ref{circuit1}. This ansatz circuit was used in all experiments. 

For each iteration, the parameters are optimized in order and only once. The optimization is sequential, starting from the top left of the circuit and moving downward in the layer. After all gates in the layer are optimized, we move on to the next layer. This process is iterated until all gates in the circuit are optimized. One trial consists of a total of 50 iterations as a stopping criterion for \verb|Rotosolve|. In terms of circuit evaluations, this corresponds to 25 and 15 iterations of optimization for \verb|Fraxis| and \verb|FQS| algorithms, respectively. The hybrid algorithms were stopped after reaching the same number of circuit evaluations as the algorithms mentioned above. This ensures a fair comparison between algorithms in terms of cost efficiency with respect to circuit evaluations performed. For the Adam algorithm, we utilize the parameter-shift rule~\cite{param_shift_rule} to compute the gradients for the parameters, which translates to 2 circuit evaluations for each gate. To ensure the same number of circuit evaluations compared to other algorithms, we set the number of iterations to 75 for the Adam algorithm. At the beginning of each trial, all parameters and gates are sampled at random with a fixed structure of the circuit. The parameters $\theta_d$ for \verb|Rotosolve| and \verb|Rotosolve-Haar| were sampled from the uniform distribution in the range of $(-\pi, \pi]$, and each unitary $V_d \in U(2)$ is sampled at random from the Haar measure. \verb|FQS| algorithm was initialized by sampling single-qubit gates from the Haar-random distribution of $U(2)$. \verb|Fraxis| algorithms components $(n_x, n_y, n_z)$ can be represented in spherical coordinates $\hat{\bm{n}}(\theta, \phi) = (\cos \theta, \sin \theta \cos \phi, \sin \theta \sin \phi)$. Here, we can create a random axis by sampling it from a uniform spherical distribution. The hybrid algorithms were initialized the same way as \verb|Rotosolve-Haar|. The Adam algorithm was initialized same way as regular \verb|Rotosolve| and the learning rate was set to 0.1 in all experiments. Data were collected by taking the expectation value of the cost function (Hamiltonian) after each gate optimization. In Sec.~\ref{overlapping_section}, a trace distance is computed instead of the expectation value of the Hamiltonian. Unless otherwise stated, all numerical results are obtained using 8192 measurement shots per Hamiltonian term. This provides a realistic estimate of measurement noise compared to near-term quantum devices.

The optimization was performed using the Pennylane package (version 0.37.0)~\cite{bergholm2022pennylaneautomaticdifferentiationhybrid} and Python version 3.11. We provide additional experiments for \verb|Rotosolve|, \verb|Fraxis|, and \verb|FQS|, where the gate optimization time and gate fidelity between the statevector simulator and with shot noise for each algorithm are presented in Appendix~\ref{gate_fidelity_and_time_appendix}.

\begin{figure}
    \centering
    \includegraphics[width=0.99\linewidth]{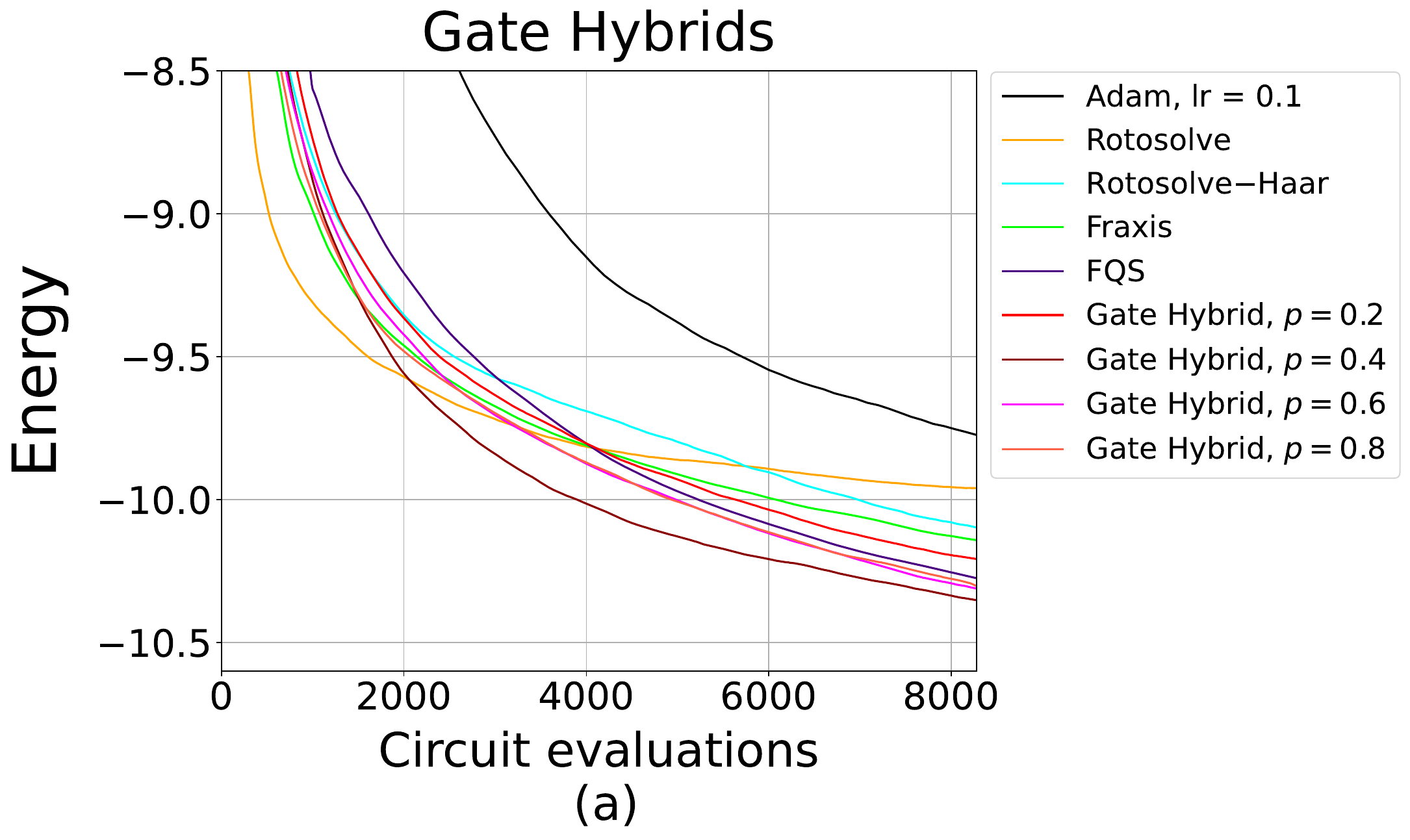}
    \includegraphics[width=0.99\linewidth]{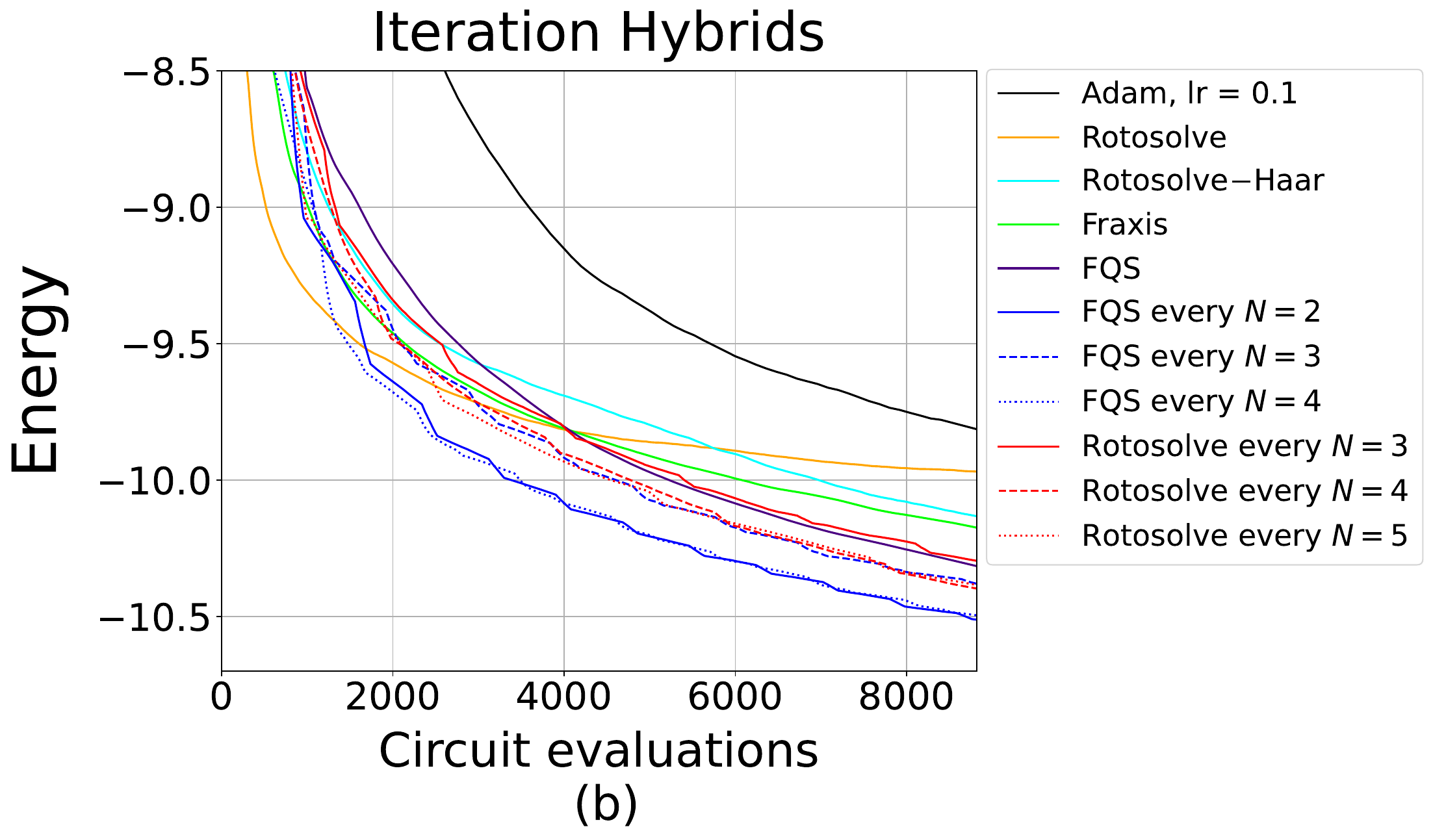}
    \caption{A comparison of algorithm performance on the 6-qubit Heisenberg model Hamiltonian with 10 layers as a function of circuit evaluations for (a) gate-specific hybrids and (b) iteration-specific hybrids. Each line shows the fitted mean of the 20 trials. The ground state of the system is approximately $E_g=-11.2111$.}
    \label{1D_heisenberg_results_6qubits}
\end{figure}

\begin{figure}
    \centering
    \includegraphics[width=0.99\linewidth]{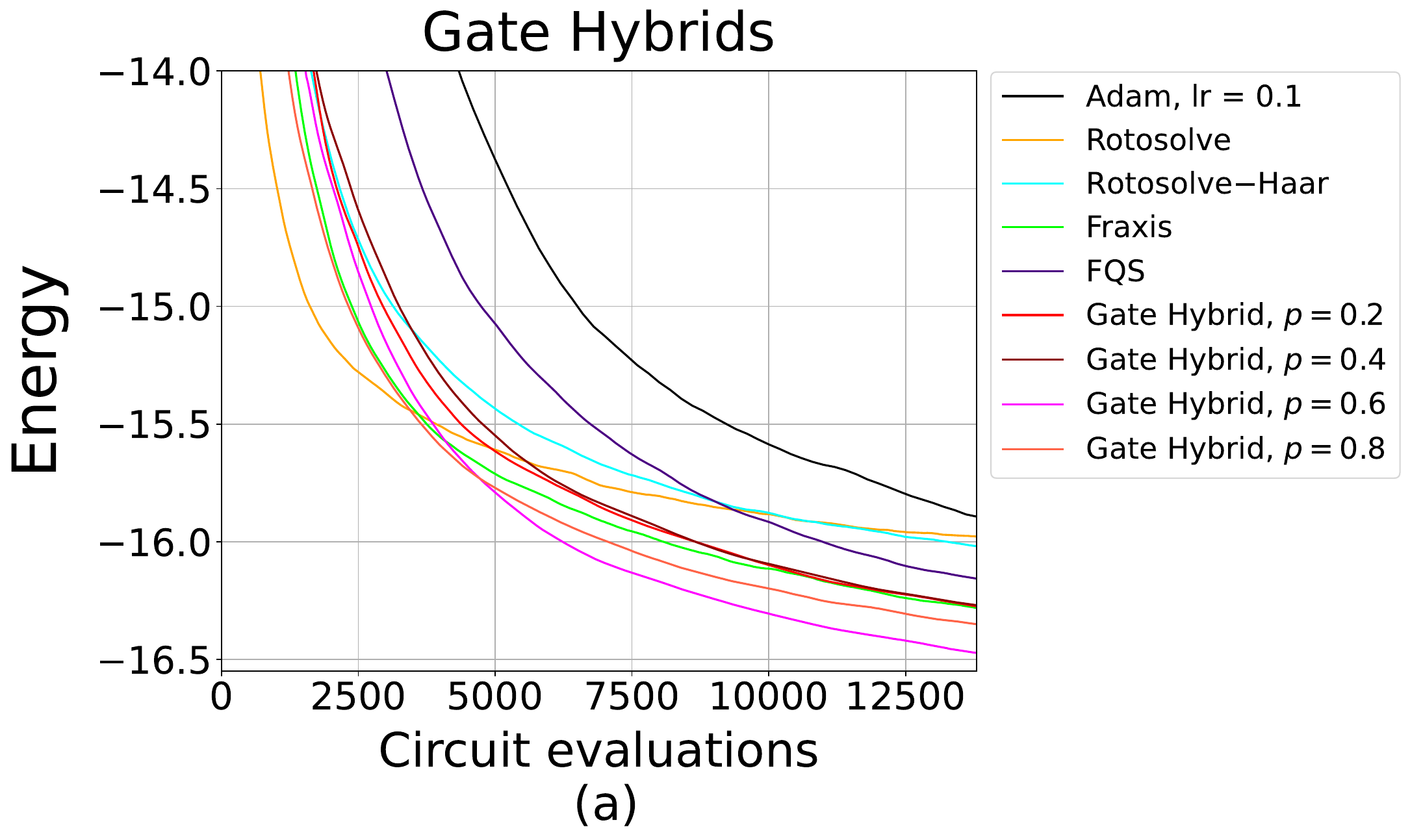}
    \includegraphics[width=0.99\linewidth]{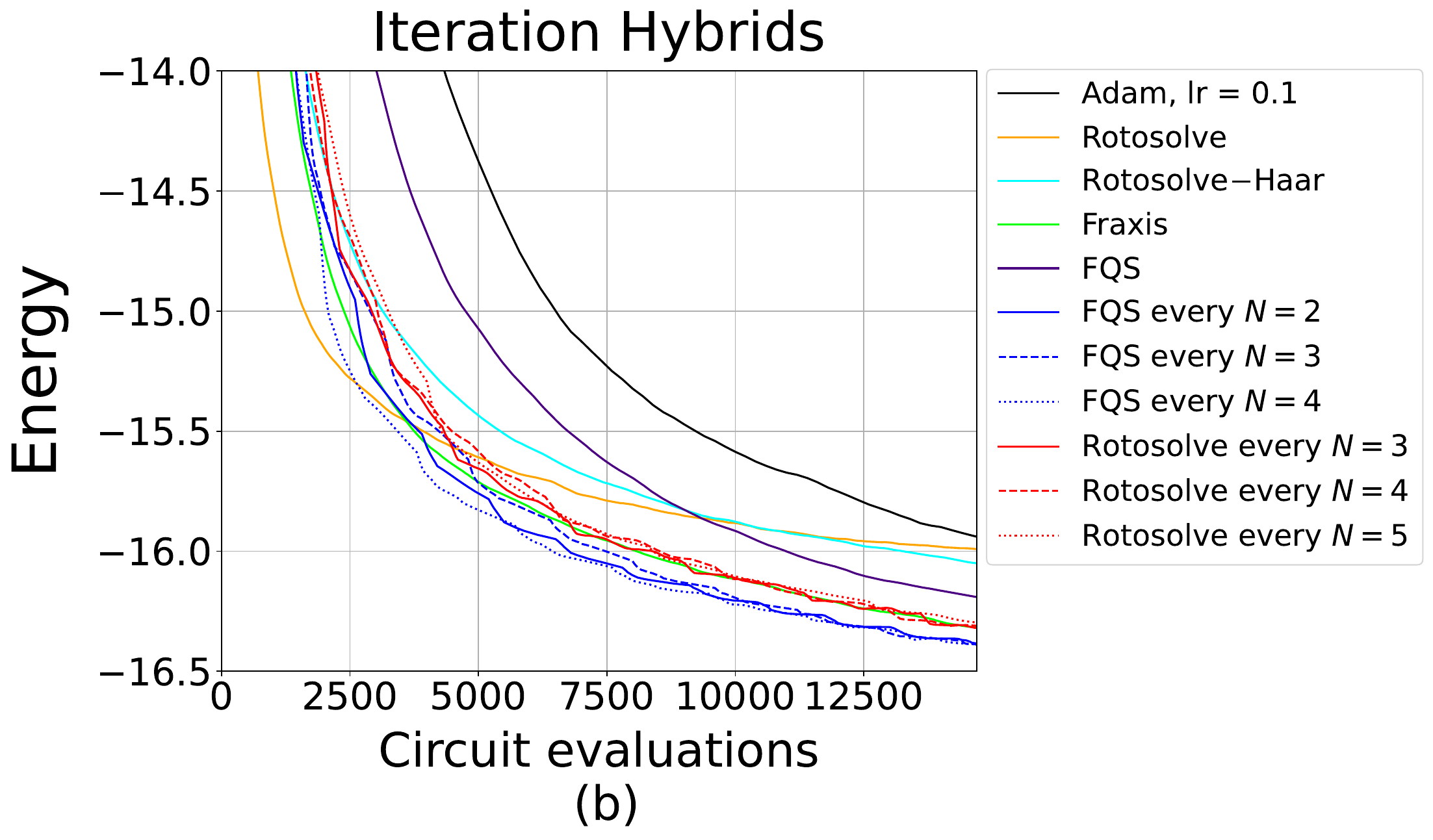}
    \caption{A comparison of algorithm performance on the 10-qubit Heisenberg model Hamiltonian with 10 layers as a function of circuit evaluations for (a) gate-specific hybrids and (b) iteration-specific hybrids. Each line shows the fitted mean of the 20 trials. The ground state of the system is approximately $E_g=-18.3688$.}
    \label{1D_heisenberg_results_10qubits}
\end{figure}

\begin{figure}
    \centering
    \includegraphics[width=0.99\linewidth]{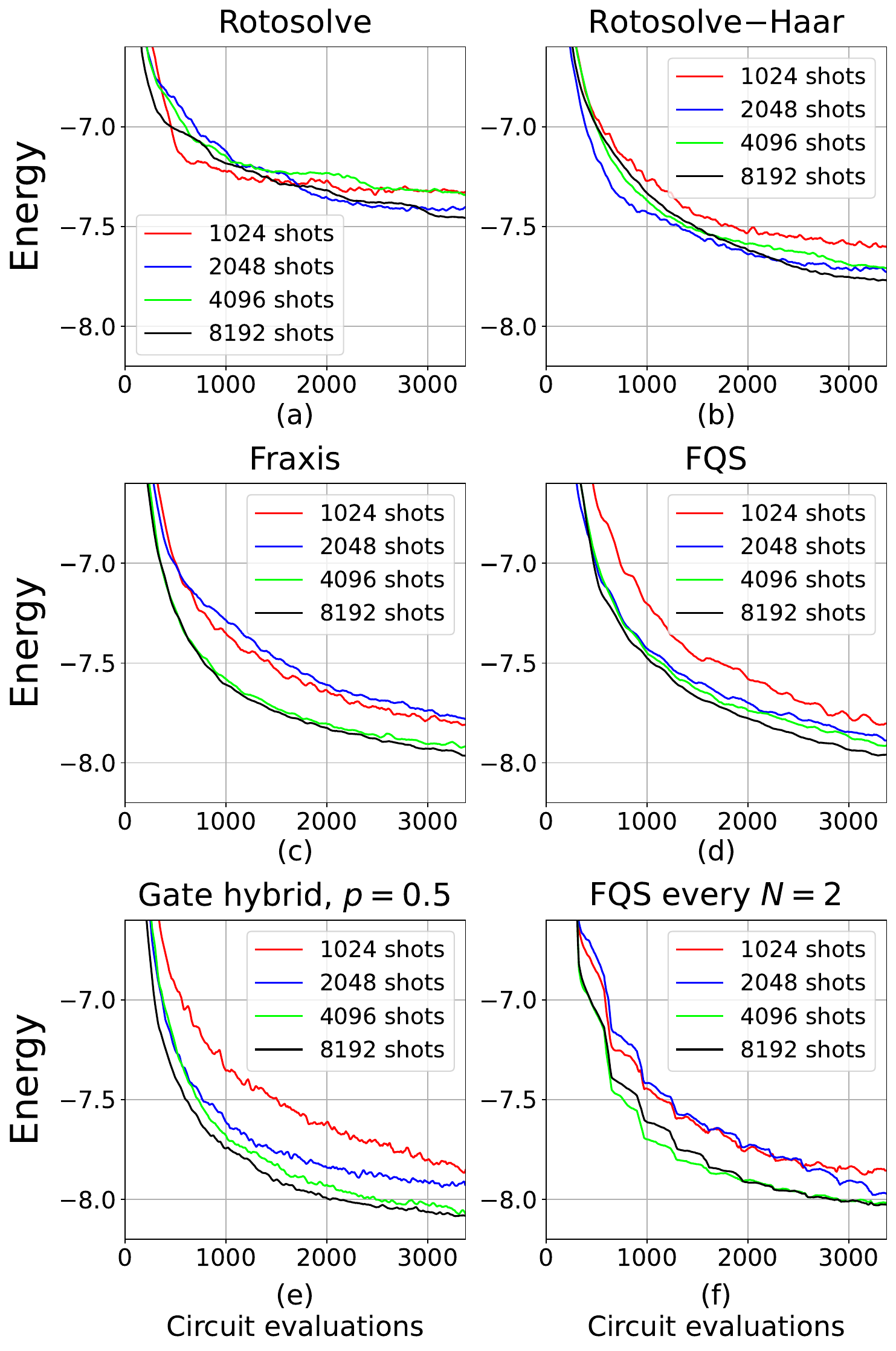}
    \caption{A comparison of algorithm performance on the 5-qubit Heisenberg model Hamiltonian as a function of circuit evaluations with $J = h = 1$. Each line shows the fitted mean of the 20 trials and represents the number of shots used to approximate each Hamiltonian term. The ground state of the system is approximately $E_g=-8.4721$.}
    \label{1D_heisenberg_5q_noisy_comparison}
\end{figure}

\subsection{Heisenberg model} \label{heisenberg_results}
The Hamiltonian of the Heisenberg model is~\cite{Kandala_2017}

\begin{equation}\label{Heisenberg_Hamiltonian}
    H = J \sum_{(i,j) \in \mathcal{E}} \left(X_i X_j + Y_i Y_j + Z_i Z_j \right) + h \sum_{i \in \mathcal{V}} Z_i,
\end{equation}
\
\\
where $\mathcal{E}$ is the edges of the lattice, $\mathcal{V}$ are the vertices, $J$ is the strength of the spin interaction and $h$ the magnetic field strength along the Z-axis. We chose $J = h = 1$ for the values $J$ and $h$ in this work.

\subsubsection{1D Heisenberg model}

The one-dimensional Heisenberg model was optimized by using Adam, \verb|Rotosolve|, \verb|Rotosolve-Haar|, \verb|Fraxis|, \verb|FQS|, and hybrid algorithms for 6 and 10 qubits. The optimization was done with 10 layers for both system sizes. We used 8192 measurements for each Hamiltonian term to approximate the energy, and a total of 20 trials were performed for each algorithm. The number of circuit evaluations is adjusted in each figure to match the true cost of each algorithm. In addition, we examined the impact on performance with different numbers of shots used in optimization. For this task, we considered a 5-qubit Heisenberg Hamiltonian with 5 layers, and the shot counts were set to 1024, 2048, 4096, 8192, respectively.

\begin{figure}
    \centering
    \includegraphics[width=0.99\linewidth]{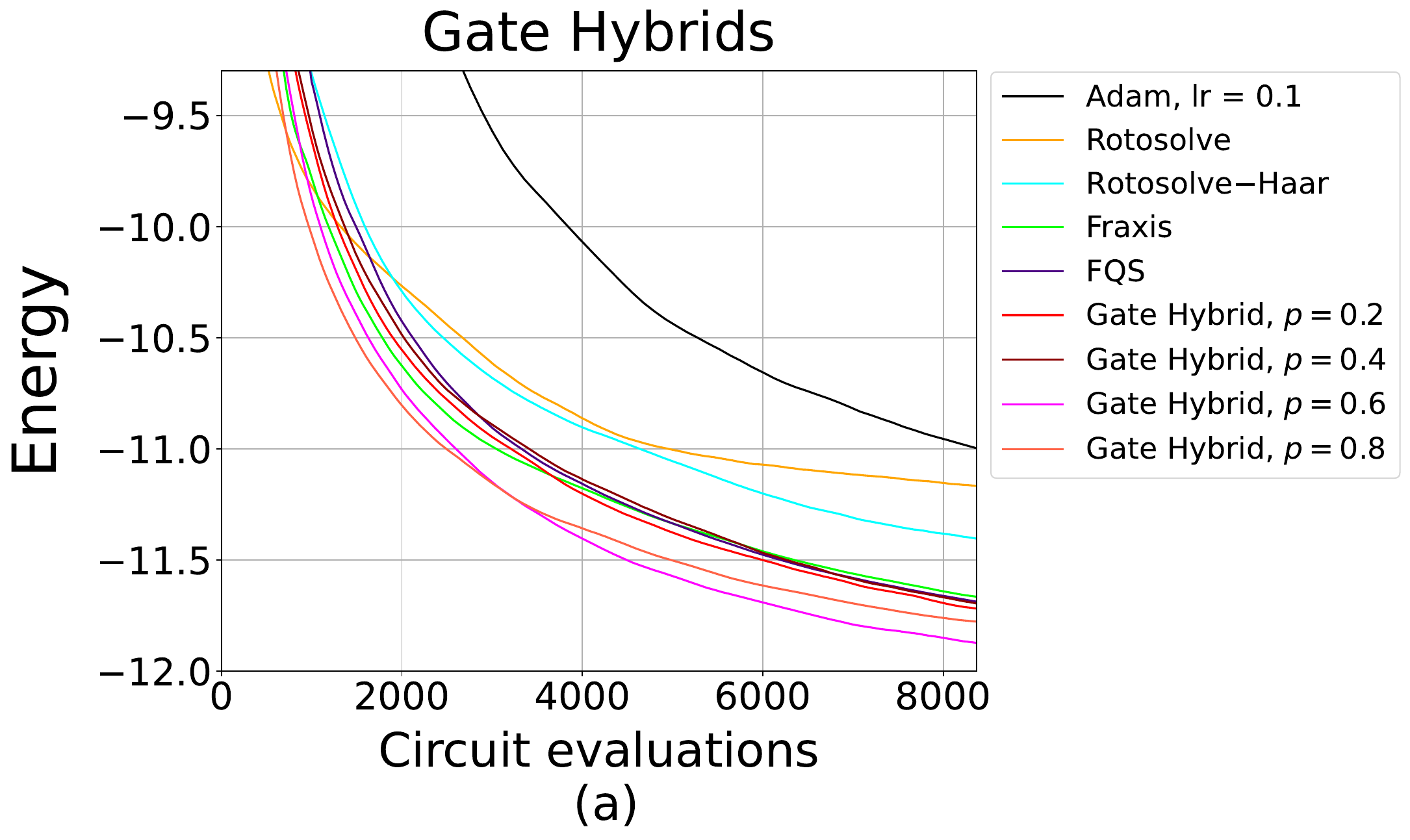}
    \includegraphics[width=0.99\linewidth]{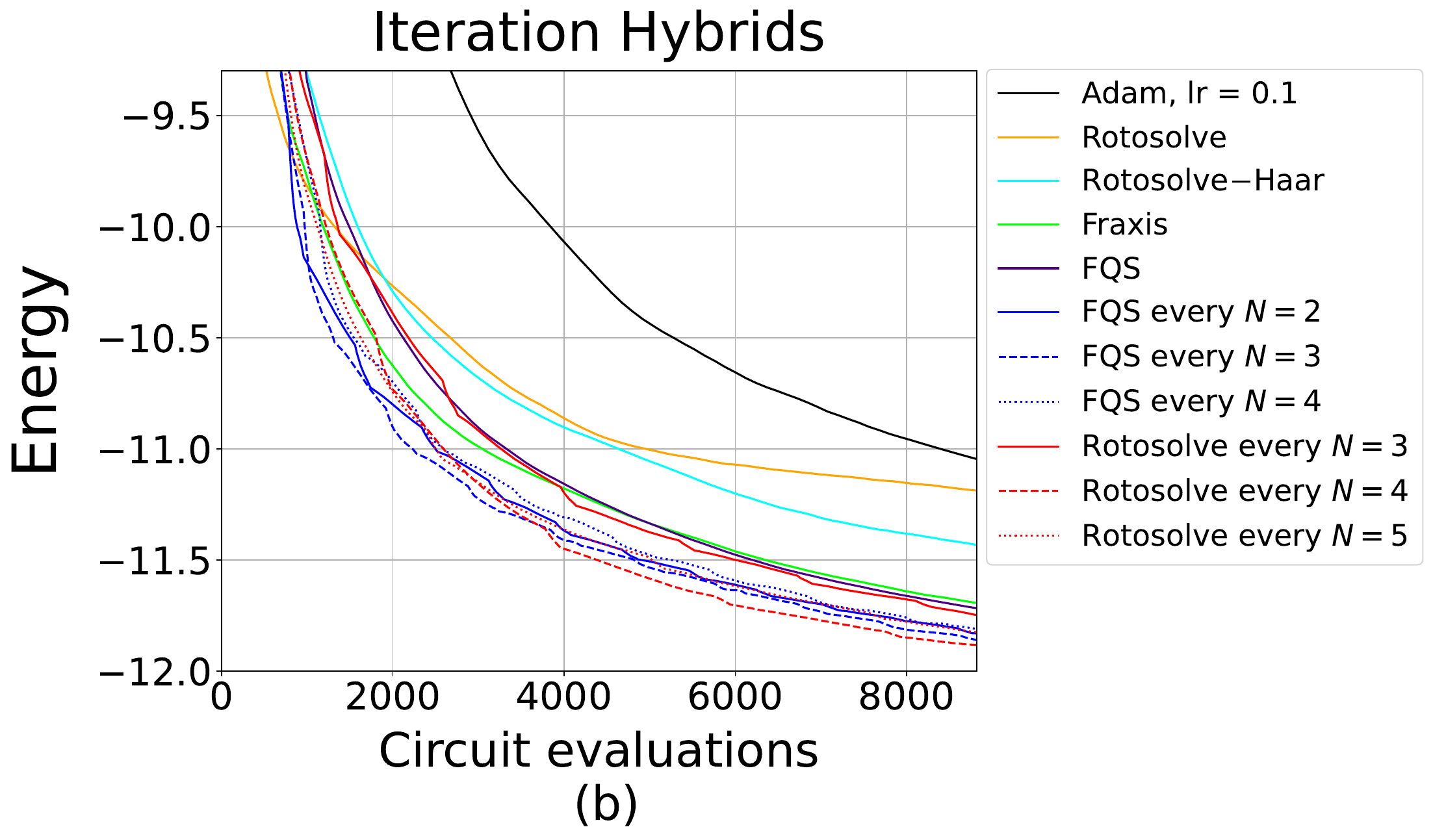}
    \cprotect\caption{A comparison of algorithm performance on the 6-qubit two-dimensional Heisenberg model Hamiltonian with 10 layers as a function of circuit evaluations for (a) gate-specific hybrids and (b) iteration-specific hybrids. Lattice dimension is set to $2\times3$ and each line represents the fitted mean of the 20 trials. The ground state of the system is approximately $E_g=-12.5175$.}
    \label{2D_heisenberg_results_6qubit}
\end{figure}

\begin{figure*}
    \centering
    \includegraphics[width=0.49\linewidth]{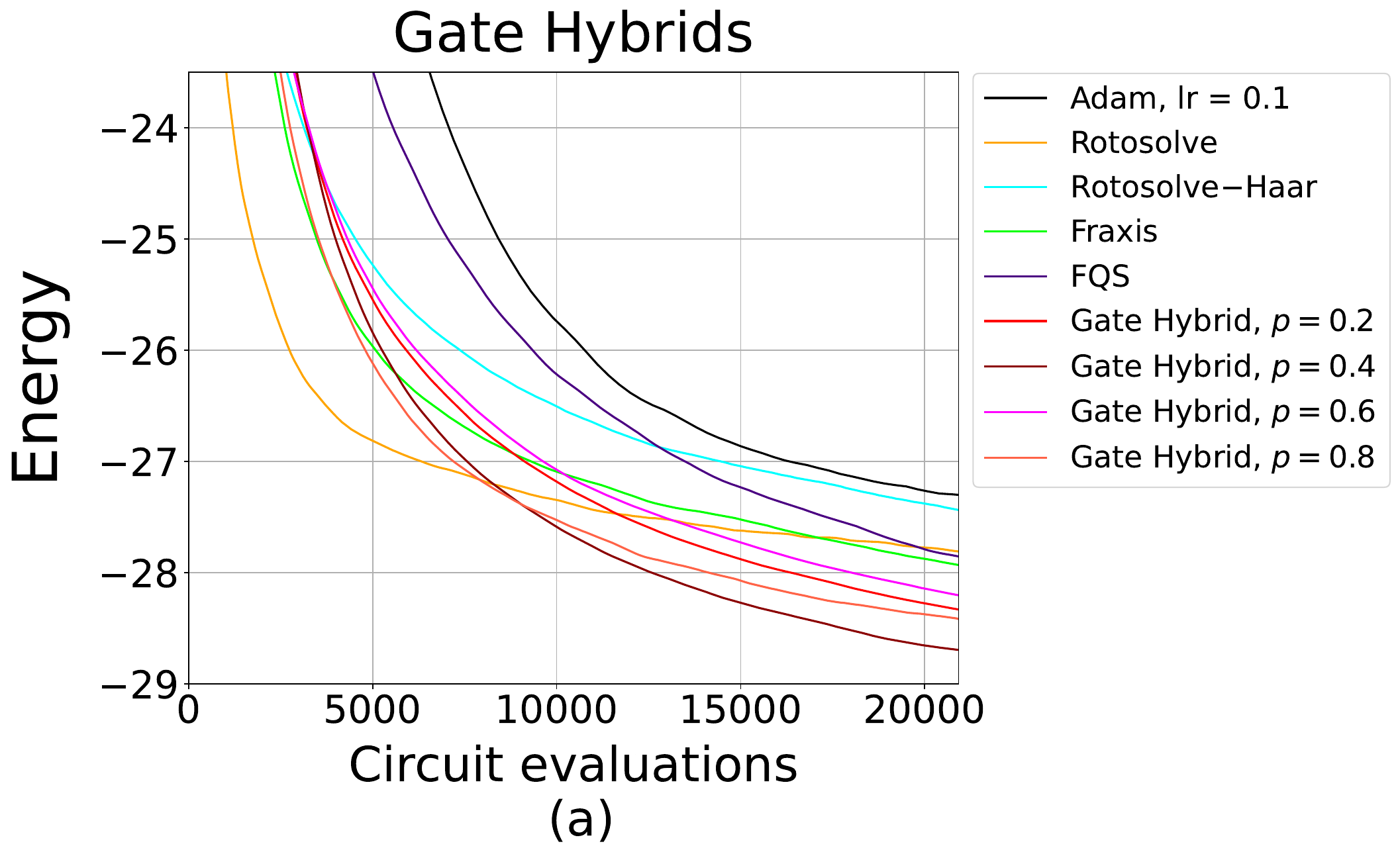}
    \includegraphics[width=0.49\linewidth]{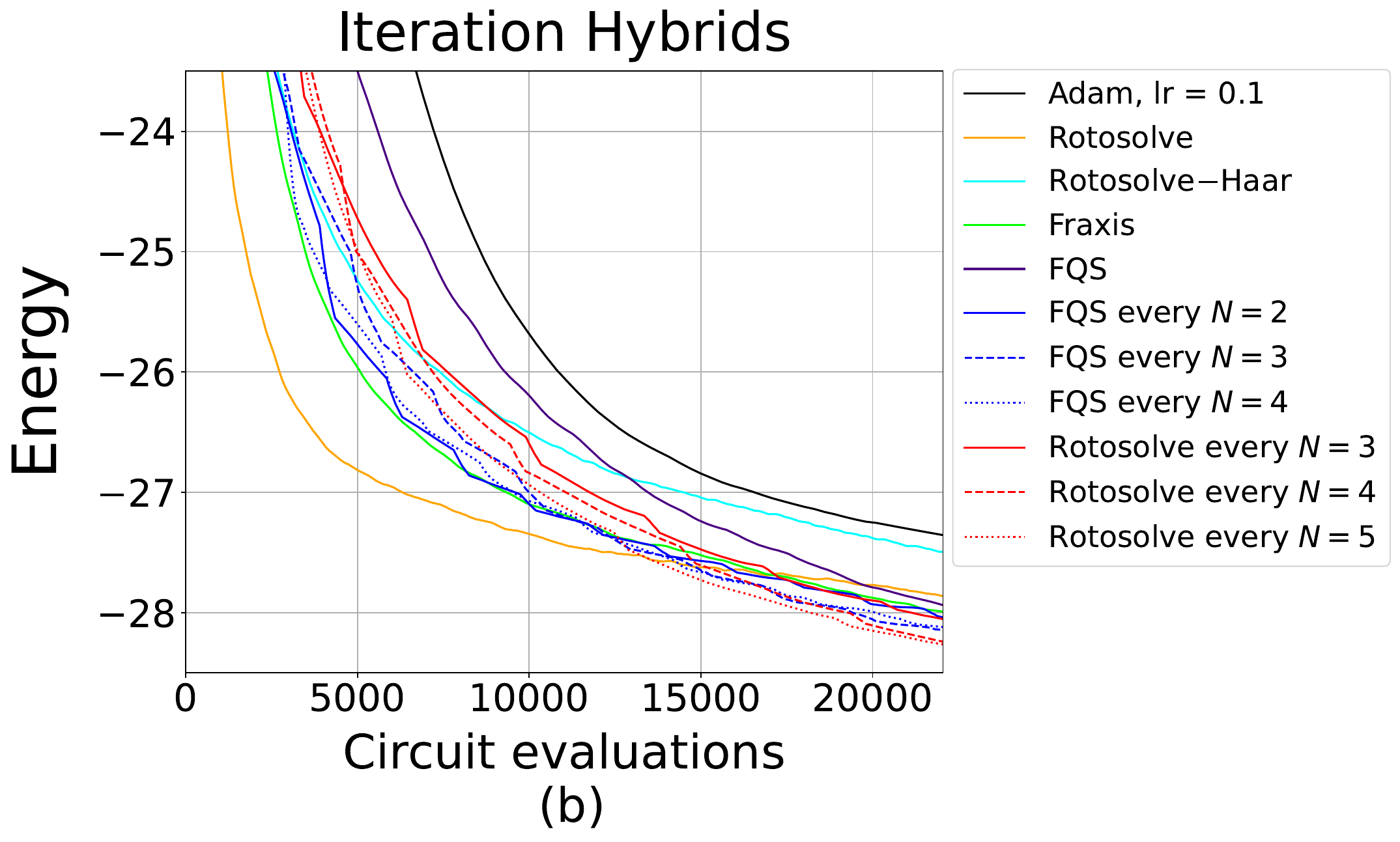}
    \cprotect\caption{A comparison of algorithm performance on the 15-qubit two-dimensional Heisenberg model Hamiltonian with 10 layers as a function of circuit evaluations for (a) gate-specific hybrids and (b) iteration-specific hybrids. Lattice dimension is set to $3\times5$, and each line represents the fitted mean of the 20 trials. The ground state of the system is approximately $E_g=-34.5505$.}
    \label{2D_heisenberg_results_15qubit}
\end{figure*}

In Fig.~\ref{1D_heisenberg_results_6qubits}(a) and Fig.~\ref{1D_heisenberg_results_6qubits}(b), we compare the performance of algorithms for the 6-qubit Heisenberg model using gate- and iteration-specific hybrids, respectively. The results demonstrate the effectiveness of hybrid algorithms compared to the known gradient-based Adam algorithm and gradient-free algorithms \verb|Rotosolve|, \verb|Fraxis|, and \verb|FQS| from Refs.~\cite{Ostaszewski_2021},~\cite{fraxis}, and~\cite{fqs_algo}, respectively, as well as the \verb|Rotosolve-Haar|. For the gate-specific hybrids, the probability $p$ was set to $p=0.2, 0.4, 0.6$, and 0.8. For the iterations-specific hybrids, we used  $N=2,3,4$ for \verb|FQS| every $N$-th iteration and values $N=3,4,5$ for \verb|Rotosolve-Haar| every $N$-th iteration. For the 6-qubit system in Fig.~\ref{1D_heisenberg_results_6qubits}, the iteration-specific hybrids exhibit the strongest performance. Here, all iteration-specific hybrids perform better than standalone \verb|FQS|, \verb|Fraxis|, and \verb|Rotosolve| algorithms, except \verb|Rotosolve-Haar| every $N=3$. The gate-specific hybrids also yield consistent improvements over baseline methods.

We repeated the experiment with the same hyperparameter values $p$ and $N$ for all gate- and iteration-specific hybrids for the 10-qubit system with 10 layers as shown in Fig.~\ref{1D_heisenberg_results_10qubits}. For the gate-specific hybrids in Fig.~\ref{1D_heisenberg_results_10qubits}(a), the best performance is obtained by focusing more on optimizing the gates with \verb|Rotosolve-Haar| than \verb|FQS|. Specifically, for $p>0.5$, the gates are optimized more frequently with \verb|Rotosolve-Haar|. When examining the iteration-specific hybrids in Fig.~\ref{1D_heisenberg_results_10qubits}(b), the \verb|FQS| every $N$-th iteration-specific hybrid achieves the best performance regardless of the hyperparameter value for $N$. In particular, the \verb|Rotosolve-Haar| every $N$-th iteration-specific hybrid exhibits similar convergence compared to \verb|Fraxis| regardless of the value $N$. These results indicate that for the 10-qubit Heisenberg model with 10 layers, it is best to use a hybrid algorithm that favors more \verb|Rotosolve-Haar| than \verb|FQS| in the optimization process.

We also investigated the effect of shot noise on the performance of gradient-free algorithms by varying the number of shots used. The results are presented in Fig.~\ref{1D_heisenberg_5q_noisy_comparison}, where the performance of each algorithm is presented for 1024, 2048, 4096, and 8192 shots, respectively. In this experiment, we considered only $p=0.5$ for the gate-specific hybrids and $N=2$ for the iteration-specific hybrid. In the comparison between \verb|Rotosolve| and \verb|Rotosolve-Haar|, \verb|Rotosolve-Haar| performs significantly better than the standard version of \verb|Rotosolve|. As the number of shots increases, i.e., more accurate measurements, the \verb|Rotosolve-Haar| exhibits consistent improvement, unlike \verb|Rotosolve|. A similar trend is observed for the gate- and iteration-specific hybrids, and they achieve the best performance with 4096 and 8192 shots.

\subsubsection{2D Heisenberg model on a rectangular lattice}

The two-dimensional rectangular Heisenberg lattice model with open boundary conditions is closely related to the one-dimensional case. The Hamiltonian for the 2D lattice is identical to the definition in Eq.~(\ref{Heisenberg_Hamiltonian}), but we have a different graph $G(V, E)$. We set parameters $J=h=1$, such that $J/h = 1$. In contrast to the one-dimensional model, we do not use the cyclic property in this case. We consider the spins to interact with their closest neighbors in a rectangular grid. 

We conducted experiments for systems of 6 and 15 qubits, using lattice dimensions of $2\times3$ and $3\times5$, respectively. In both cases, the number of layers was set to $L=10$, and a total of 20 trials were performed for each algorithm. Again, we used the total number of circuit evaluations as a metric to measure the performance of each algorithm. We considered $N=2,3,4$ for the \verb|FQS| every $N$-th iteration-specific hybrid, $N=3,4,5$ for the \verb|Rotosolve-Haar| every $N$-th iteration-specific hybrid, and $ p=0.2, 0.4, 0.6, 0.8$ for the gate-specific hybrid algorithm.

\begin{figure}
    \centering
    \includegraphics[width=0.99\linewidth]{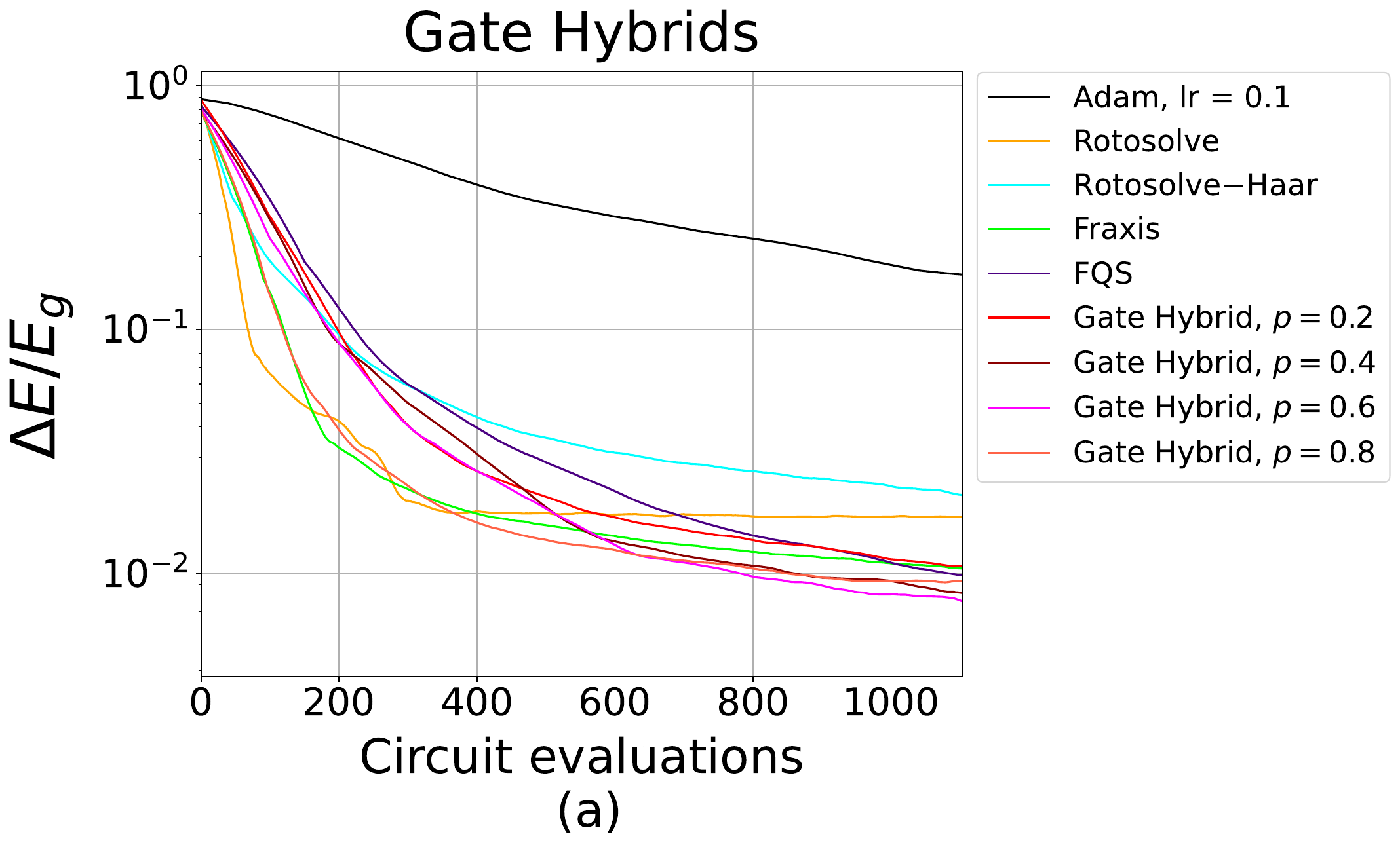}
    \includegraphics[width=0.99\linewidth]{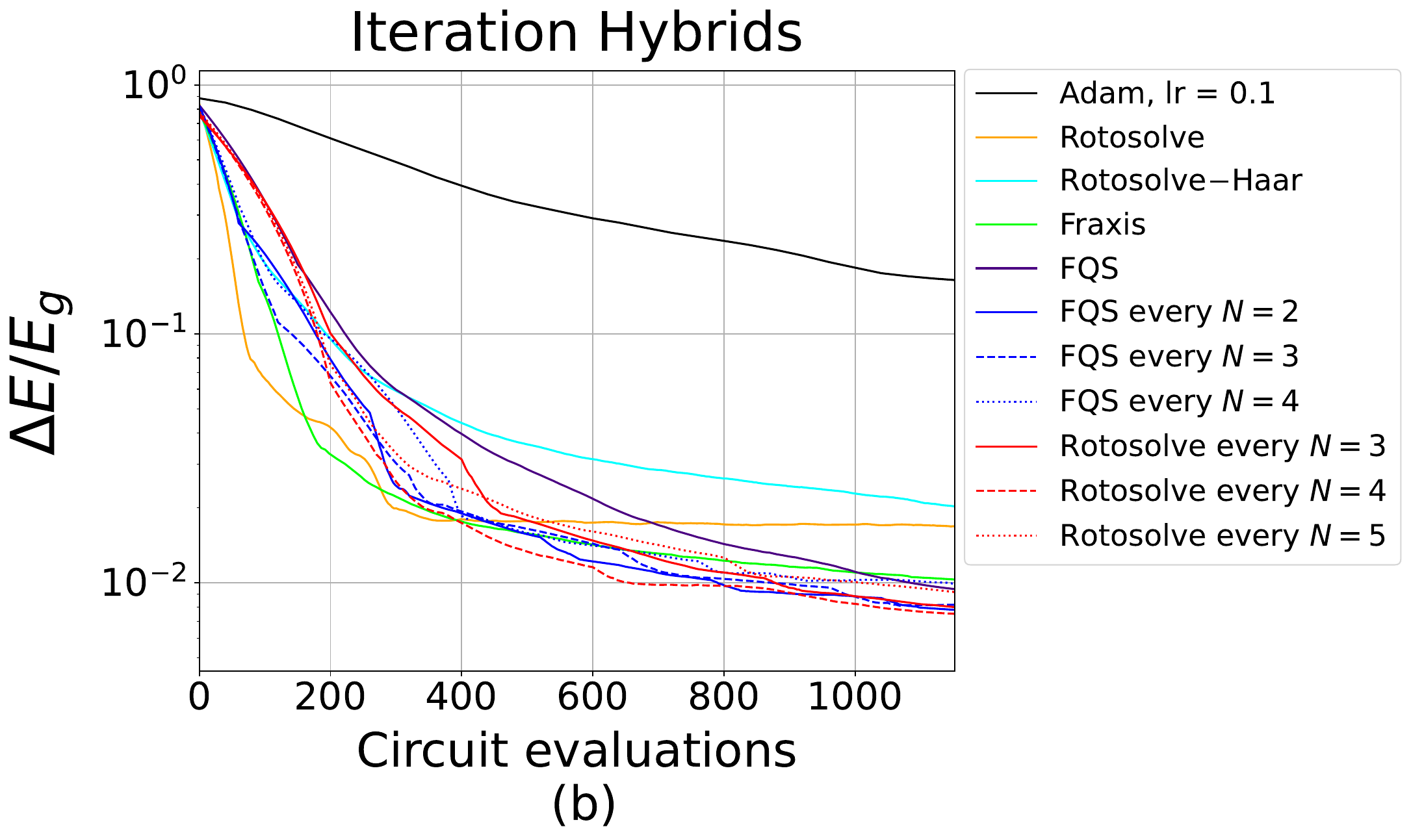}
    \cprotect\caption{A comparison of algorithm performance on the 4-qubit hydrogen molecular Hamiltonian as a function of circuit evaluations for (a) gate-specific hybrids and (b) iteration-specific hybrids. Each line is the mean of 20 trials, and the ground state of the Hamiltonian is approximately $E_g=-1.1373$.}
    \label{H2_results}
\end{figure}

The results for 6- and 15-qubit systems are presented in Fig.~\ref{2D_heisenberg_results_6qubit} and Fig.~\ref{2D_heisenberg_results_15qubit}, respectively. We observe that the gate-specific hybrids perform better for the 15-qubit system than the 6-qubit system. In Fig.~\ref{2D_heisenberg_results_15qubit}(a), regardless of the hyperparameter value $p$, the gate-specific hybrid achieves the best convergence among the standalone algorithms. In Fig.~\ref{2D_heisenberg_results_15qubit}(b), all except \verb|FQS| every $N=2$ and \verb|Rotosolve-Haar| every $N=3$ iteration-specific hybrids outperform the other algorithms. When comparing the relative performance of the iteration-specific hybrids for the 6-qubit system in Fig.~\ref{2D_heisenberg_results_6qubit}(b), all of them have better convergence except \verb|Rotosolve-Haar| every $N=3$ iteration-specific hybrid. Based on these observations, it is good to choose a gate-specific hybrid, while setting the value of $p > 0.5$ is a favorable choice for the circuit depth of $L=\mathcal{O}(n)$. The iteration-specific hybrids appear to be a beneficial choice for smaller system sizes, but with suitable hybrid and hyperparameter selection, they are also a viable option for larger and more complex systems. In addition, \verb|Rotosolve| benefits from the random axis initialization at smaller system sizes as demonstrated for 6-qubit systems for one- and two-dimensional Heisenberg models in Fig.~\ref{1D_heisenberg_results_6qubits} and Fig.~\ref{2D_heisenberg_results_6qubit}, respectively.

\subsection{Hydrogen molecular Hamiltonian} \label{results_molecular_hamiltonians}

Next, we consider the optimization of the molecular hydrogen (H$_2$) Hamiltonian using 4 qubits. The Hamiltonian for H$_2$ consists of 15 Pauli terms, which are listed in the Appendix~\ref{hydrogen_hamiltonian_appendix}. The Hamiltonian was extracted from the Pennylane Quantum Datasets for the H$_2$ molecule~\cite{Utkarsh2023Chemistry_H2}, and the bond length was set to 0.742 Å, and the STO-3G basis set~\cite{sto_3g_basis} was used. As in the previous experiments,  we examine the convergence of the energy expectation value as a function of the number of circuit evaluations.

A total of 20 trials were performed, and 8192 measurements for each Hamiltonian term were used to estimate the energy. We used the same layer design as in Fig.~\ref{circuit1} but for 4 qubits, and the number of layers was set to $L=5$. Data were collected in each gate optimization step, and the energy expectation value was measured. 

The results are shown in Fig.~\ref{H2_results}. Here, we examine the relative error with respect to the ground state $E_g$ of the H$_2$ Hamiltonian as a function of circuit evaluations on a semi-log scale. We denote the relative error to the ground state as $\Delta E/ E_g$, where $\Delta E = E_g - \expval{M}$. In Fig.~\ref{H2_results}(a), the gate-specific hybrids with $p=0.4$ and $p=0.6$ perform the best, indicating that choosing $p$ close to 0.5 effectively leverages the fast convergence and cost effectiveness of \verb|Rotosolve-Haar|. For the iteration-specific hybrids in Fig.~\ref{H2_results}(b), the best convergence towards the ground state is obtained with a low value for the hyperparameter $N$ for both iteration-specific hybrids. Notably, the convergence of the best iteration-specific hybrids shows no significant change from each other, but they still perform better than standalone algorithms.

\subsection{Maximizing overlap between the quantum states} \label{overlapping_section}

Maximizing the overlap between quantum states is a fundamental task in quantum computing. The overlap or closeness between quantum states can be quantified by various measures, with two of the most common being the trace distance and fidelity. These measures provide different but related ways to assess how similar or "close" two quantum states are to each other. Quantifying the overlap between quantum states is crucial in various applications, from quantum state preparation to quantum error correction. For two pure states $|\phi\rangle,|\psi\rangle$ we can express the fidelity as 
\begin{equation}
    F(|\phi\rangle,|\psi\rangle) = |\langle\phi|\psi\rangle|^2,
\end{equation}
while the trace distance is given by 
\begin{equation}
    T(|\phi\rangle,|\psi\rangle) = \sqrt{ 1 - F(|\phi\rangle,|\psi\rangle)}.
\end{equation}
In this work, we use the trace distance as a metric to measure the proximity of quantum states. There are numerous applications for maximizing the overlap between quantum states, namely in quantum state tomography~\cite{quantum_state_tomography}, quantum state preparation~\cite{quantum_state_prep}, and quantum error correction~\cite{gottesman2009introductionquantumerrorcorrection}. 

\begin{figure}
    \centering
    \includegraphics[width=0.99\linewidth]{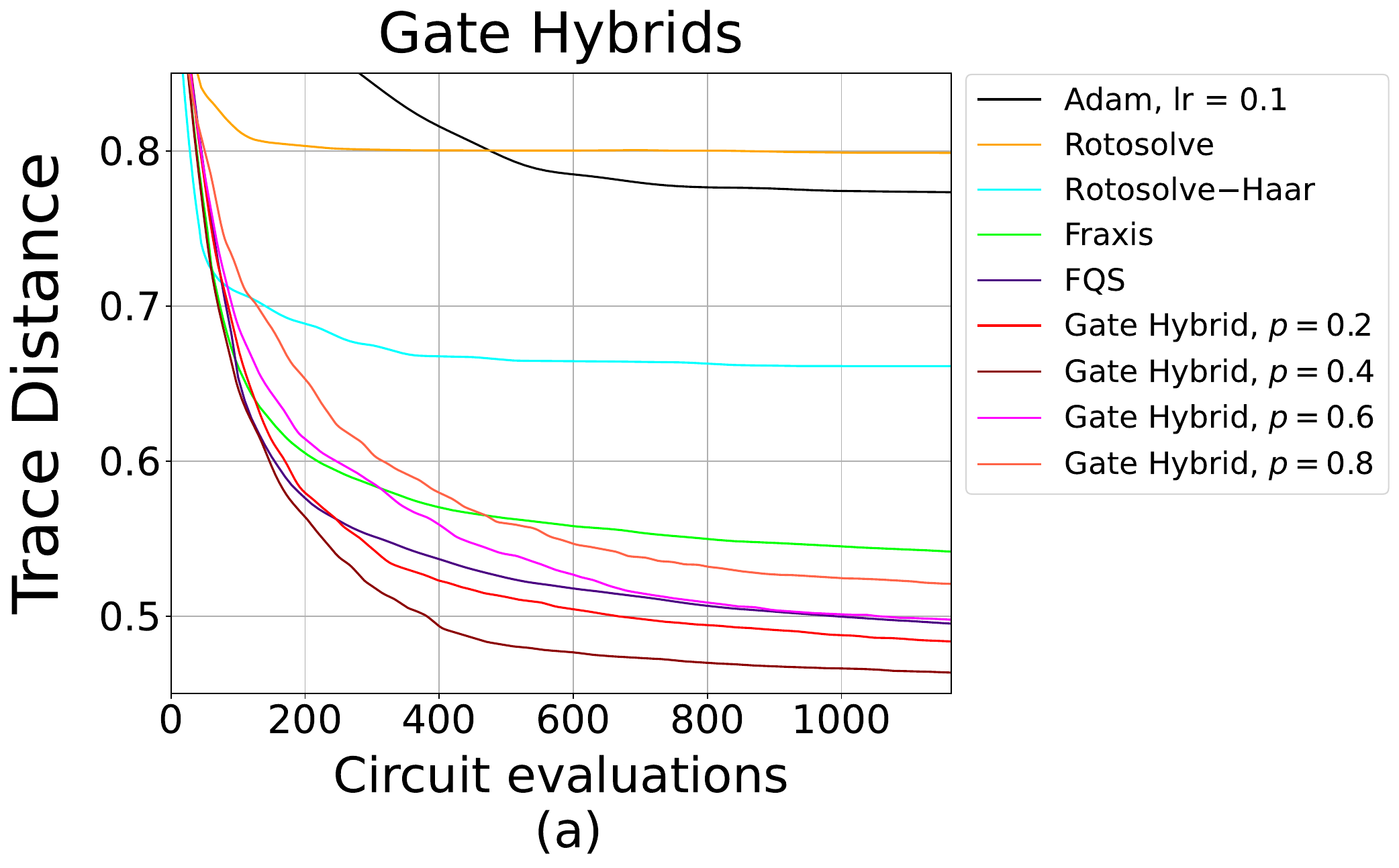}
    \includegraphics[width=0.99\linewidth]{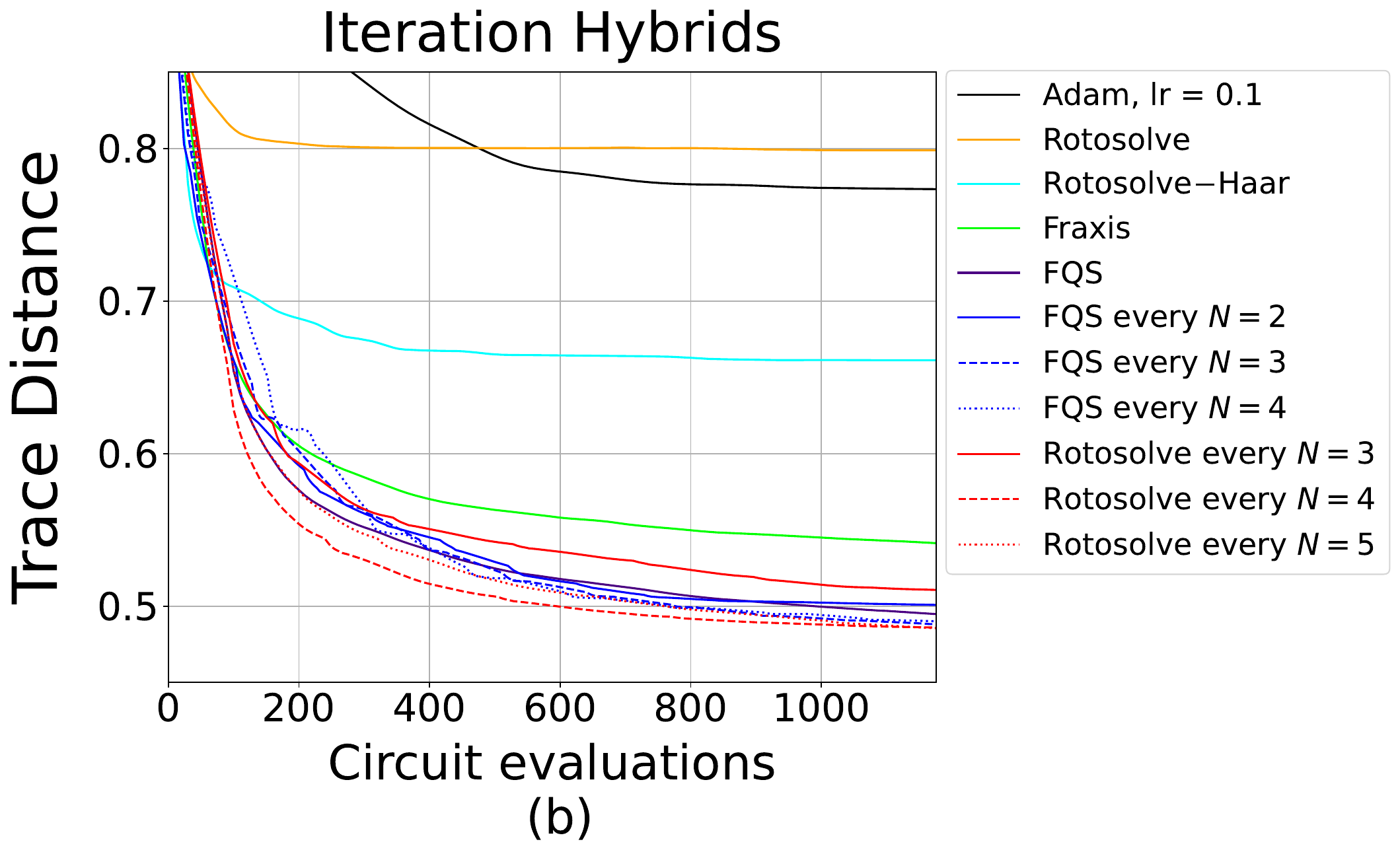}
    \cprotect\caption{Maximizing overlapping between a random quantum state $|\phi \rangle$ and VQC by using trace distance for algorithms Adam, \verb|Rotosolve|, \verb|Rotosolve-Haar|, \verb|Fraxis|, \verb|FQS| and (a) gate-specific hybrid as well as (b) iteration-specific hybrids. Each plot shows the trace distance as a function of circuit evaluations. In each trial, a new 4-qubit random state was generated. Each line indicates the mean of 30 trials, and the lower value is better.}
    \label{overlapping_4q}
\end{figure}

\begin{figure}
    \centering
    \includegraphics[width=0.99\linewidth]{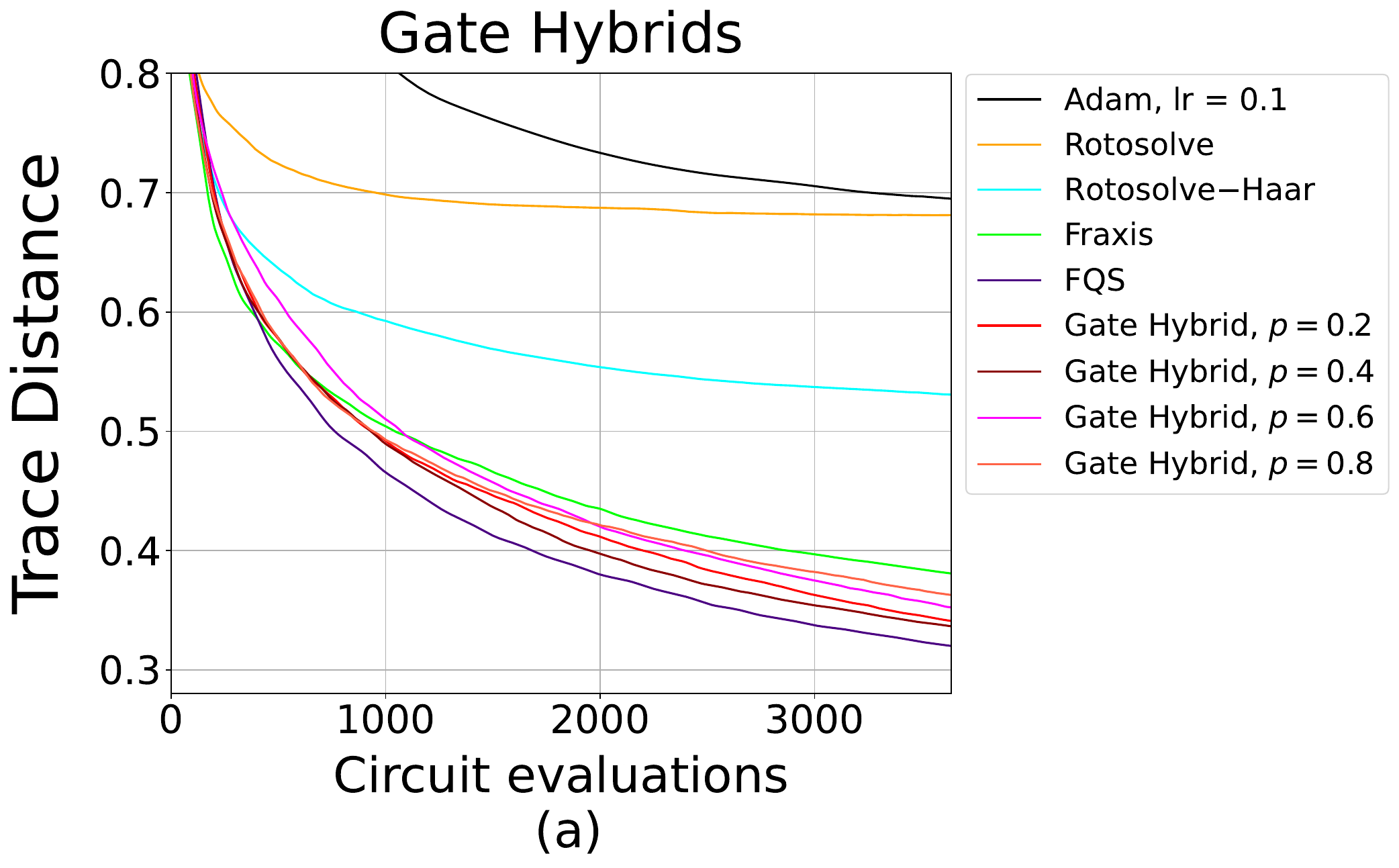}
    \includegraphics[width=0.99\linewidth]{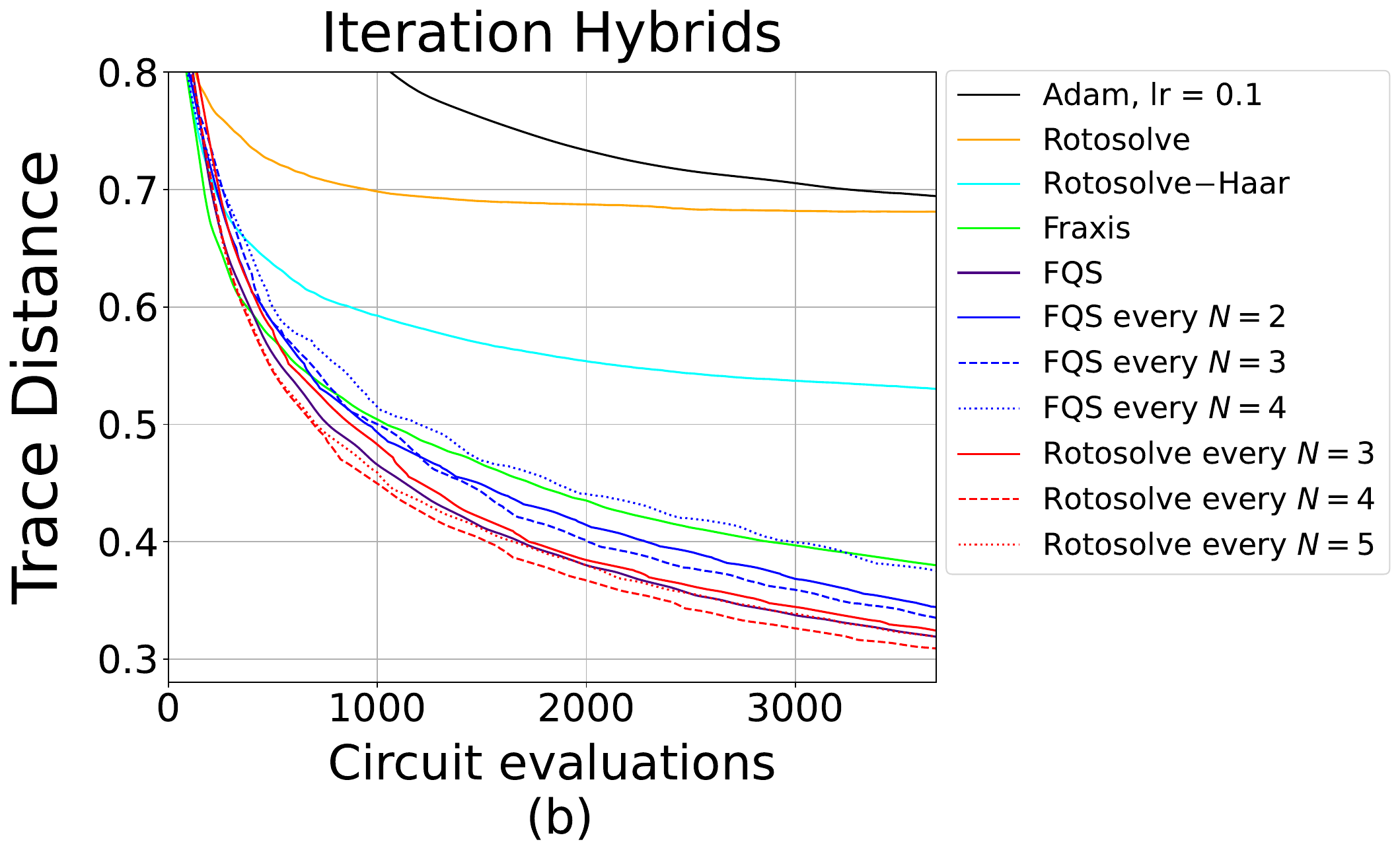}
    \cprotect\caption{Maximizing overlapping between a random quantum state $|\phi \rangle$ and VQC by using trace distance for Adam, \verb|Rotosolve|, \verb|Rotosolve-Haar|, \verb|Fraxis|, and \verb|FQS| algorithms as well as (a) gate-specific hybrid and (b) iteration-specific hybrids. Each plot shows the trace distance as a function of circuit evaluations. In each trial, a new 5-qubit random state was generated. Each line indicates the mean of 30 trials, and the lower value is better.}
    \label{overlapping_5q}
\end{figure}

We consider Adam, \verb|Rotosolve|, \verb|Rotosolve-Haar|, \verb|Fraxis|, \verb|FQS|, and hybrid algorithms to maximize the overlap between a random quantum state $| \phi \rangle$ and the state produced by VQC as closely as possible. To maximize the overlap, we define a projection operator $M = - | \phi\rangle \langle \phi|$, which we use as the Hermitian observable in our cost function. Denoting the output state from the parametrized circuit as $|\psi_{\boldsymbol{\theta}}\rangle$, the cost function to be minimized can then be expressed in terms of the fidelity or the trace distance as 
\begin{equation}
    \langle M \rangle_{\boldsymbol{\theta}} = -F(|\phi\rangle, |\psi_{\boldsymbol{\theta}}\rangle) = T(|\phi\rangle, |\psi_{\boldsymbol{\theta}}\rangle)^2 - 1.
\end{equation}
This optimizes the VQC to a random quantum state $|\phi \rangle$. We use the same ansatz circuit as in Fig.~\ref{circuit1} for 4 qubits with 2 layers, and 5  qubits with 5 layers, respectively. For each combination of qubits and the number of layers in the circuit, we conducted a total of 30 independent trials. We used 50 iterations of optimization for \verb|Rotosolve| and adjusted the iterations for other algorithms based on the total number of circuit evaluations. This translates to 75 iterations for Adam, 25 iterations for \verb|Fraxis|, and 15 iterations for \verb|FQS|, respectively. The hybrid algorithms were stopped after reaching the maximum number of circuit evaluations. In each trial a new random quantum state $|\phi\rangle$, which was generated as follows: First, we draw $2^N$ random complex numbers $z_k = a_k + b_k i$ from the standard normal distribution $\mathcal{N}(0,1)$ to create a vector $|\psi \rangle = (z_1,\dots,z_{2^N})$. We sample real and imaginary parts $a_k$ and $b_k$ separately. After that, we normalize the vector to have unit length $|\phi\rangle = |\psi\rangle / || |\psi\rangle ||^2$. Finally, we create the Hermitian operator $M = - |\phi \rangle \langle \phi|$, which we use in optimization as the objective function.

The results of the random state optimization for each algorithm are shown in Fig.~\ref{overlapping_4q} and Fig.~\ref{overlapping_5q} for the 4- and 5-qubit systems, respectively. Overall, the results favor the algorithms with higher expressivity, reflecting the hardness of random quantum state preparation. In this case, it is better to choose the algorithm with the most expressivity, i.e. \verb|FQS| algorithm. For the 4-qubit system with 2 layers shown in Fig.~\ref{overlapping_4q}(a), the gate-specific hybrid approach with $p=0.2$ and $p=0.4$ slightly outperforms the \verb|FQS| algorithm, despite being the most expressive standalone algorithm. The iteration-specific hybrids in Fig.~\ref{overlapping_4q}(b) similarly obtain slightly better convergence than \verb|FQS|. To be precise, \verb|Rotosolve-Haar| every $N$ and \verb|FQS| every $N$ with values $N=4,5$ and $N=3,4$ exhibit a comparable behavior for convergence, both slightly surpassing \verb|FQS|. 

However, for the 5-qubit system with 5 layers in Fig.~\ref{overlapping_5q}, the convergence of each gate-specific hybrid regardless of the hyperparameter $p$ lies between \verb|FQS| and \verb|Fraxis|. The best convergence among the gate-specific hybrids is obtained for smaller values of $p$. From the iteration-specific hybrids, only \verb|Rotosolve-Haar| every $N=4$ has marginally better convergence than \verb|FQS|, and other iteration-specific hybrids lie between \verb|FQS| and \verb|Fraxis|. We note that the iteration-specific hybrid that uses \verb|Rotosolve-Haar| every $N$-th iteration performs better in this task than the iteration-specific hybrid that uses \verb|FQS| every $N$-th iteration. This behavior arises from the complex nature of the cost function for which the algorithm with the highest expressivity is better suited. In addition, for both 4- and 5-qubit systems, \verb|Rotosolve-Haar| has significantly better performance than regular \verb|Rotosolve|, highlighting the benefit of the random axis initialization when applying it to complex cost functions.

\section{Conclusions} \label{section_conclusion}

In several cases, we observe that hybrid algorithms consisting of \verb|Rotosolve-Haar| and \verb|FQS| improve the convergence towards the ground state of the Hamiltonian as a function of circuit evaluations. First, a gate-specific hybrid algorithm was proposed, where with a fixed probability $p$ we optimize the current gate with \verb|Rotosolve-Haar| and otherwise \verb|FQS| is used. Next, an iteration-specific hybrid algorithm was proposed. Here, the hybrid algorithm is composed of two optimization algorithms $\mathcal{A}$ and $\mathcal{B}$, e.g., \verb|Rotosolve-Haar| and \verb|FQS|, respectively. The algorithm $\mathcal{A}$ is used with every $N$-th iteration, and otherwise, the algorithm $\mathcal{B}$ is used. In this work, we constructed two iteration-specific hybrids, using \verb|Rotosolve-Haar| and \verb|FQS| for the algorithms $\mathcal{A}$ and $\mathcal{B}$, respectively, and vice versa. 

Then, we examined how the performance of the proposed hybrid algorithms compared to other VQAs in optimization. We found that for the 6- and 10-qubit one-dimensional Heisenberg models, the best choice is a hybrid algorithm, iteration- or gate-specific. We also found that for the 5-qubit system, the iteration- and gate-specific hybrids performed the best when increasing the shot count used to estimate the Hamiltonian terms. In addition, \verb|Rotosolve-Haar| outperformed regular \verb|Rotosolve| in this task for all shot counts used. Then, we experimented on the two-dimensional Heisenberg Hamiltonian on a rectangular lattice with 6 and 15 qubits, reaching the same conclusion as in the one-dimensional Heisenberg model. Optimizing the 4-qubit molecular hydrogen (H$_2$) Hamiltonian also showed good results for the hybrid algorithms, but the random state optimization, on the other hand, yielded more nuanced results. When we used 4 qubits with 2 layers, the results were better with certain hyperparameters $p$ and $N$, whereas with 5 qubits and 5 layers, only \verb|Rotosolve-Haar| every $N=4$ iteration-specific hybrid achieved better convergence than \verb|FQS|. The gate-specific hybrids, on the other hand, outperform \verb|Fraxis| but do not surpass \verb|FQS| in this case. A consistent performance hierarchy was observed among the algorithms in the random-state optimization task, with \verb|Rotosolve| or Adam as the worst, followed by \verb|Rotosolve-Haar| and \verb|Fraxis|, respectively. After that, the next best were the iteration-specific or gate-specific hybrids, depending on the hyperparameters $p$ and $N$, and finally the \verb|FQS|. We note that the iteration-specific hybrids exhibit strong overall performance in random state optimization.

When evaluating iteration-specific hybrids with different values of $N$, we found that the \verb|FQS| every $N$-th iteration-specific hybrid achieves the best performance at $N=2$, when simpler Hamiltonians are considered, such as the one-dimensional Heisenberg model. For larger or complex cost functions, including maximizing the overlap between the quantum state, one should focus on using the iteration-specific hybrid that uses \verb|Rotosolve-Haar| every $N$-th iteration. This iteration-specific hybrid uses more \verb|FQS| in the optimization process than the \verb|Rotosolve-Haar| algorithm. Across our experiments, we identified that $N=4$ is the most effective for more complex cost functions.

The hybrid algorithms proposed in this work are a viable alternative for VQC optimization alongside highly expressive \verb|FQS|. The advantages of the hybrid algorithms arise from a suitable trade-off between the expressivity and the optimizability of the circuit ansatz. The hybrids are well-suited for optimizing different spin Hamiltonians with various system sizes. When dealing with a smaller number of qubits or shallow VQCs, for the gate-specific hybrid, a smaller probability value $p<0.5$ is generally more preferable. Similarly, when the system is completely random, or we have an unknown Hermitian observable, as in Sec.~\ref{overlapping_section}, the use of a small value of $p$ remains the most effective choice for the gate-specific hybrid approach. For the iteration-specific hybrids, setting $N=2$ yields the most consistent performance when optimizing spin Hamiltonians up to 10 qubits. For larger or more complex systems, the \verb|Rotosolve-Haar| every $N$-th iteration-specific hybrid benefits setting $N>2$. Moreover, the random axis initialization for \verb|Rotosolve|, \verb|Rotosolve-Haar| is well suited for the optimization tasks with more complex cost functions, such as a random target quantum state. \verb|Rotosolve-Haar| is also a good option when optimizing smaller spin Hamiltonians like 5- or 6-qubit one-dimensional Heisenberg models.

We leave open the question of how this random initialization method, as well as the hybrid algorithm methods, could be utilized on $n$-qubit gates, where $n \geq 2$. More sophisticated combinations of algorithms and hybrid strategies may be explored further than presented in this work, which we plan to investigate in future studies.

\begin{acknowledgments}
We acknowledge funding by Business Finland for project 8726/31/2022 CICAQU. J.P. received funding from the InstituteQ doctoral school. 
\end{acknowledgments}

\appendix

\section{Coefficients for the Pauli terms of the molecular Hamiltonians} \label{hydrogen_hamiltonian_appendix}

Here, we present the Hamiltonian terms constituting the molecular Hamiltonian for the Hydrogen molecule H$_2$. Hamiltonian terms are expressed in a tensor product basis where terms that commute with each other form a group, i.e., IZIZ and ZIZI, IIIZ are grouped since they commute. This eases the readability of the Hamiltonian terms. We denote I, X, Y, Z as the identity operator and Pauli matrices $\sigma_x, \sigma_y$ and $\sigma_z$ respectively. In total, there were 15 Pauli terms for the Hydrogen molecule H$_2$ with 4 qubits. All Pauli terms and their corresponding coefficients are represented in Table~\ref{h2_table}.

\begin{table}[h] 
\centering
\scalebox{0.85}{%
\begin{tabular}{ |c|c|c| } 
 \multicolumn{3}{c}{} \\
 \hline
 IIII  & YXXY & XXYY \\ 
 $-0.09963387941370971$ & 0.04533062254573469 & $-0.04533062254573469$  \\ 
 ZIII &  XYYX & YYXX \\ 
 0.17110545123720233 & 0.04533062254573469  & $-0.04533062254573469$ \\
 IZII & &  \\ 
 0.17110545123720233 &  & \\ 
 ZZII &  & \\ 
 0.16859349595532533  & &  \\
 IIZI &  &  \\
 $-0.22250914236600539$ &  & \\
 ZIZI &  & \\ 
 0.12051027989546245 & &  \\ 
 IIIZ &  & \\ 
 $-0.22250914236600539$  & &  \\ 
 ZIIZ &  & \\ 
 0.16584090244119712 & &   \\ 
 IZZI & &  \\ 
 0.16584090244119712  &  & \\ 
 IZIZ &  & \\ 
 0.12051027989546245   &  & \\ 
 IIZZ &  & \\ 
 0.1743207725924201 &  & \\

 \hline
\end{tabular}}
\caption{Pauli terms and their coefficients for H$_2$ molecular Hamiltonian at bond distance 0.742 Å.}
\label{h2_table}
\end{table}

\section{Proof for unitary representation} \label{appendix_rotosolve_haar_proof}

Here we provide a proof for solving the equation $U = \exp(-i \frac{\theta}{2} V H V^\dagger)$ w.r.t. angle $\theta$ and unitary $V$.

\begin{proposition}
    Let $H\in \mathbb{C} ^{2\times 2}$ be a Hermitian matrix with eigenvalues $-1$ and $1$, and let $U\in SU(2)$.
    \begin{enumerate}
        \item There exists a unitary matrix $V\in U(2)$ and a number $\theta\in \mathbb{R}$ such that
        \begin{equation}
            \label{eq:V-theta-U-equality}
            e^{-i\frac{\theta}{2} VHV^\dagger} = U.             
        \end{equation}
        \item Suppose that $U\neq\pm I$. If $(V_j,\theta_j)\in U(2)\times \mathbb{R}$, $j=1,2$, are such that
        \begin{equation}
            \label{eq:two-thetas-Vs}
            e^{-i\frac{\theta_1}{2} V_1HV_1^\dagger} = U = e^{-i\frac{\theta_2}{2} V_2HV_2^\dagger},        
        \end{equation}
        then
        \begin{equation}
            \label{eq:uniqueness-of-orbit}
            \{e^{-i\frac{\theta}{2} V_1HV_1^\dagger} : \theta\in \mathbb{R} \}
            = \{e^{-i\frac{\theta}{2} V_2HV_2^\dagger} : \theta\in \mathbb{R} \}.
        \end{equation}
    \end{enumerate}
\end{proposition}
\begin{proof}
    Let $\gamma\in \mathbb{R}$ be such that the eigenvalues of $U$ are $e^{\pm i\gamma}$, and let $(\ket{u_+},\ket{u_-})$ and $(\ket{h_+},\ket{h_-})$ be orthonormal bases for $\mathbb{C}^2$  such that the spectral decompositions for $U$ and $H$ are
    \begin{equation*}
        U = e^{i\gamma}\ketbra{u_+}{u_+} + e^{-i\gamma}\ketbra{u_-}{u_-}
    \end{equation*}
    and
    \begin{equation*}
        H = \ketbra{h_+}{h_+} - \ketbra{h_-}{h_-}.
    \end{equation*}
    Then 
    \begin{equation*}
        e^{-i\frac{\theta}{2} VHV^\dagger} = e^{-i\frac{\theta}{2}}V\ketbra{h_+}{h_+}V^\dagger + e^{i\frac{\theta}{2}}V\ketbra{h_-}{h_-}V^\dagger,
    \end{equation*}
    so~\labelcref{eq:V-theta-U-equality} holds if, for example, $\theta=-2\gamma$ and $V\in U(2)$ is such that $V\ket{h_\pm}=\ket{u_\pm}$.

    Assume then that $U\neq\pm I$. Let $(V_j,\theta_j)\in U(2)\times\mathbb{R}$, $j=1,2$, be such that~\labelcref{eq:two-thetas-Vs} holds, and write $P_\pm := V_1\ketbra{h_\pm}{h_\pm}V_1^\dagger$ and $Q_\pm := V_2\ketbra{h_\pm}{h_\pm}V_2^\dagger$. Then
    \begin{equation*}
        e^{-i\frac{\theta_1}{2}} P_+ + e^{i\frac{\theta_1}{2}} P_-
        = e^{-i\frac{\theta_2}{2}} Q_+ + e^{i\frac{\theta_2}{2}} Q_-,
    \end{equation*}
    where  $e^{-i\frac{\theta_1}{2}}\neq e^{i\frac{\theta_1}{2}}$ and $e^{-i\frac{\theta_2}{2}}\neq e^{i\frac{\theta_2}{2}}$.

    The uniqueness of the spectral decomposition implies that either $P_\pm=Q_\pm$, in which case $V_1HV_1^\dagger=V_2HV_2^\dagger$, or $P_\pm=Q_\mp$, in which case $V_1HV_1^\dagger=-V_2HV_2^\dagger$. In either case~\labelcref{eq:uniqueness-of-orbit} follows.
\end{proof}

To solve the angle $\theta$ and the unitary matrix $V$, we first compose a unitary matrix $U$ from the quaternion representation. We note that 

\begin{equation}
R(\bm{q}) = \bm{q}\cdot \vec{\varsigma} = q_0 I - i(q_1\sigma_x + q_2 \sigma_y + q_3 \sigma_z),
\end{equation}
\\
from which we obtain the matrix form as follows:

\begin{equation}
    U = R(\bm{q}) = 
    \begin{pmatrix} 
    q_0 - iq_3  & -q_2 - iq_1 \\
    q_2 - iq_1  & q_0 + iq_3
    \end{pmatrix}.
\end{equation}\\
Then, we solve the eigenvectors and eigenvalues for $U$ as well as the eigenvectors for the matrix $H$ (i.e., $\sigma_z$). After that, we form new matrices $S$ and $Q$ whose columns are the eigenvectors of $U$ and $H$, respectively.  Finally, we calculate the unitary matrix $V$ as the product of $S$ and $Q$ and angle $\theta$ from the first eigenvalue $\lambda$ of the matrix $U$

\begin{align}
    V &= S Q^\dagger, \\[0.2cm]
    \frac{\theta}{2} &= \arctan\left( \frac{\Im(\lambda)}{\Re(\lambda)}\right).
\end{align}

\begin{figure*}
    \centering
    \includegraphics[width=0.99\linewidth]{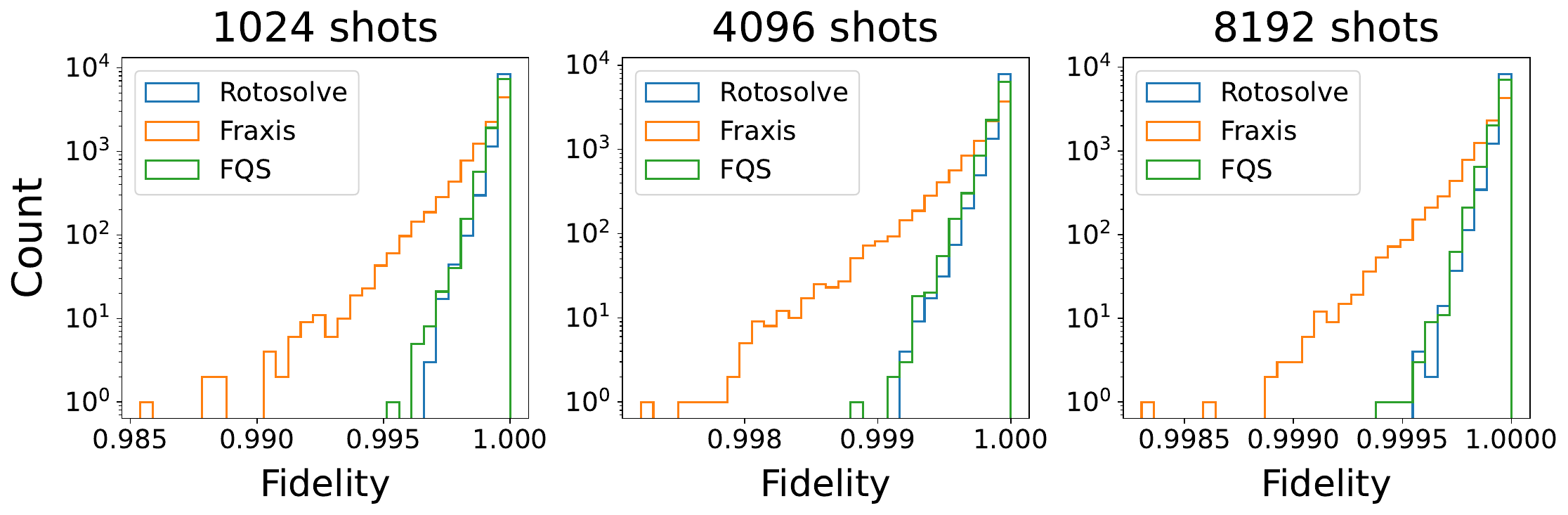}
    \cprotect\caption{Gate fidelities with 1024, 4096, and 8192 counts of shot noise compared to gate optimization done with the statevector simulator for \verb|Rotosolve| (blue), \verb|Fraxis| (orange), and \verb|FQS| (green) algorithms. Each histogram consists of $10^4$ trials.}
    \label{gate_fidelity}
\end{figure*}


\section{Gate Fidelities and Execution Times}\label{gate_fidelity_and_time_appendix}

\begin{figure}
    \centering
    \includegraphics[width=0.99\linewidth]{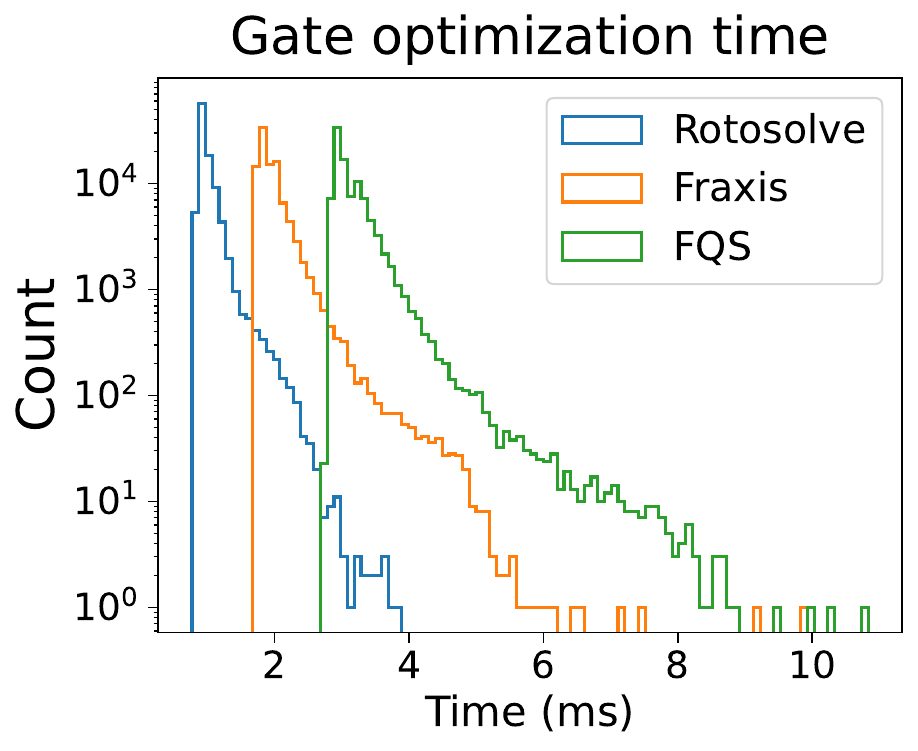}
    \cprotect\caption{Gate execution times for \verb|Rotosolve| (blue), \verb|Fraxis| (orange), and \verb|FQS| (green) algorithms performed with the statevector simulator. Each histogram denotes a distribution of $10^5$ trials for the corresponding algorithm and is composed of 100 bins. The vertical axis is the count of each bin, and the horizontal axis is the time of gate optimization in milliseconds.}
    \label{gate_execution_time}
\end{figure}

\begin{table}
\centering
\scalebox{1.3}{%
\begin{tabular}{|c|c|c|c|} 
 \hline
 Algorithm & Rotosolve  & Fraxis & FQS \\ 
 \hline\hline
 mean (ms) & 1.0020 & 1.9815 &  3.1687 \\ 
 \hline
 std (ms) & 0.1864 & 0.3066 & 0.4012 \\
 \hline
 var (ms) & 0.0348 &  0.0940 & 0.1610 \\
 \hline
\end{tabular}}
\caption{Mean gate parameter initialization and optimization times for each algorithm, as well as standard deviation and variance in milliseconds.} 
\label{gate_execution_table}
\end{table}

We analyze the single-qubit gate fidelities obtained using the statevector simulator and with different numbers of shots for the \verb|Rotosolve|, \verb|Fraxis|, and \verb|FQS| algorithms. Accounting for the shot noise in the measurements provides a quantitative measure of how robust the algorithms are against the shot noise. We compare the gate fidelity between the optimization performed with the statevector simulator and with the shot noise. The parameters of each algorithm were kept identical for both the statevector simulator and when accounting for the shot noise. We consider a single-qubit Hermitian operator defined as the sum of Pauli matrices, $H=X + Y + Z$, and evaluate the gate fidelities using 1024, 4096, and 8192 shots. A total of $10^4$ trials were performed, and in each trial, the initial parameters were drawn from a uniform distribution for the corresponding algorithm as described at the beginning of Sec.~\ref{results_section}. Then the optimization is performed separately using the statevector simulator and simulations with shot noise. Finally, the fidelity between the optimized single-qubit gates is computed.

The results for each shot count are presented in Fig.~\ref{gate_fidelity}. Across all shot counts \verb|Rotosolve| and \verb|FQS| obtain similar gate fidelity distributions, whereas \verb|Fraxis| has a broader distribution of gate fidelities compared to the result obtained with a statevector simulator. The gate fidelity improves with increasing shot count, corresponding to more accurate measurements. Nevertheless, the overall trends remain largely unchanged.

Additionally, we evaluated the computational time required to optimize the single gate with \verb|Rotosolve|, \verb|Fraxis|, and \verb|FQS| algorithms with the statevector simulator. The distribution of $10^5$ trials for each algorithm is shown in Fig.~\ref{gate_execution_time}, while the numerical mean, standard deviation, and variance are provided in Table~\ref{gate_execution_table}. The optimization was performed using the Pennylane \verb|lightning.qubit| device, which enables fast statevector simulation. We used an Intel i7-14700F CPU, and the execution times were measured using the \verb|timeit| Python package. As one might expect, the execution times are proportional to the number of measurements required to optimize a single gate. From Table~\ref{gate_execution_table}, the gate execution time for \verb|Fraxis| is, on average, twice as much as that of \verb|Rotosolve|. This coincides in terms of required circuit evaluations to optimize the gate, as \verb|Rotosolve| requires 3 circuit evaluations and \verb|Fraxis| requires 6 circuit evaluations. Comparing \verb|Rotosolve| to \verb|FQS|, we nearly have the same ratio with respect to gate execution times and the required circuit evaluations. However, in Fig.~\ref{gate_execution_time}, the \verb|FQS| exhibits a broader distribution of gate execution times than \verb|Fraxis| or \verb|Rotosolve|, respectively.

\newpage

\bibliography{apssamp}

\end{document}